\documentclass[showpacs,preprintnumbers,amsmath,amssymb,12pt]{article}
\usepackage{a4wide,amsmath,amsthm,epsfig,graphicx}
\usepackage{amsmath,amsthm,amssymb}
\usepackage{amsfonts}
\usepackage{graphics}
\usepackage{dsfont}
\usepackage{bbm}
\usepackage{bm}
\usepackage{xcolor}
\usepackage{dcolumn}
\usepackage{fancyhdr,fancybox}
\usepackage[svgnames]{pstricks}
\usepackage{pst-text}
\usepackage{pst-plot,pst-eucl,pstricks-add}
\usepackage{caption}
\usepackage{subcaption}
\usepackage{systeme, mathtools}
\usepackage{algorithm}
\usepackage{algpseudocode}

\usepackage{slashbox}

\algrenewcommand\alglinenumber[1]{{\sffamily\footnotesize#1}}
\algnewcommand\Not{\textbf{not}}
\algnewcommand\FALSE{FALSE}
\algnewcommand\TRUE{TRUE}

\newtheorem{thm}{Theorem}[section]
\newtheorem{cor}[thm]{Corollary}

\newtheorem{prop}[thm]{Proposition}

\newtheorem{remark}{Remark}

\newcommand{\refeqq}[1]{~(\ref{#1})}
\newcommand{\myref}[1]{~\ref{#1}}
\newcommand{\mycite}[1]{~\cite{#1}}

\newcommand{\R}{\mathbb{R}}

\newcommand{\rv}{\textit{rv}}
\newcommand{\id}{\textit{id}}
\newcommand{\iid}{\textit{iid}}
\newcommand{\sd}{\textit{sd}}

\newcommand{\LKh}{L\'{e}vy-Khintchin}

\newcommand{\pdf}{\textit{pdf}}

\newcommand{\chf}{\textit{chf}}
\newcommand{\cgf}{\textit{cgf}}

\newcommand{\Levy}{L\'{e}vy}

\newcommand{\Pqo}{\bm{P}\hbox{-\emph{a.s.}}}
\newcommand{\eqd}{\stackrel{d}{=}}

\newcommand{\EXP}[1]{\bm{E}\left[{#1}\right]}

\newcommand{\lch}{\textit{lch}}

\newcommand{\poiss}{\mathcal{P}}

\newcommand{\unif}{\mathcal{U}}

\newcommand{\gam}{\mathcal{G}}

\newcommand{\arem}{$a$-remainder}

\newcommand{\cts}{\mathcal{CTS}}
\newcommand{\bcts}{\mathcal{BCTS}}

\usepackage[a4 paper, top = 1.3 cm, bottom = 1.5 cm, left = 1.3 cm, right = 1.3 cm]{geometry}
\title{\Huge \textbf{Pricing Energy Derivatives in Markets Driven by Tempered Stable and CGMY Processes of Ornstein-Uhlenbeck Type
}}

\author{Piergiacomo \textsc{Sabino}\footnote{piergiacomo.sabino@eon.com} \footnote{The views, opinions, positions or strategies expressed in this article are those of the author
and do not necessarily represent the views, opinions, positions or strategies of, and should not be
attributed to E.ON SE.}\\
Quantitative Risk Management\\
E.ON SE\\
\vspace{5pt}
 Br\"{u}sseler Platz 1, 45131 Essen, Germany
}

\date{}
\begin{document}
    \maketitle \thispagestyle{empty}
        \begin{abstract}
				In this study we consider the pricing of energy derivatives when the evolution of spot prices follows a tempered stable or a CGMY driven Ornstein-Uhlenbeck process. To this end, we first calculate the characteristic function of the transition law of such processes in closed form.  This result is instrumental for the derivation of non-arbitrage conditions such that the spot dynamics is consistent with the forward curve. Moreover, based on the results of Cufaro Petroni and Sabino\mycite{cs20_3}, we also conceive efficient algorithms for the exact simulation of the skeleton of such processes and propose a novel procedure when they coincide with compound Poisson processes of Ornstein-Uhlenbeck type. We illustrate the applicability of the theoretical findings and the simulation algorithms in the context of the pricing different contracts namely, strips of daily call options, Asian options with European style and swing options. Finally, we present an extension to future markets.
        \end{abstract}
        \section{Introduction}\label{sec:introduction}
Most energy and commodity markets exhibit seasonality, mean-reversion high volatilities and occasional distinctive price spikes,
which results in demand for derivative products which protect the holder against
high prices. In equity markets there is a clear evidence that asset returns are not Gaussian and it is common practice to rely on \Levy\ processes, other than the Brownian motion, in order to capture heavy-tails and jumps of the log-prices. Several empirical studies (see for instance Carr and Crosby\mycite{CarrCrosby10}) have shown that CGMY \Levy\ processes introduced by Carr et al.\mycite{CGMY2002}, named after its authors, are a valuable alternative. Moreover, the class of such processes is quite flexible and also encompasses  Variance Gamma processes introduced in Madan and Seneta\mycite{MadanSeneta90} and is on the other hand, a special case of the wider class of bilateral tempered stable processes (see K\"{u}chler and Tappe\mycite{KT2013}).  

Commodity and energy markets however, exhibit mean-reversion which cannot be described by plain \Levy\ processes but rather by \Levy-driven Ornstein-Uhlenbeck (OU) processes. Based on these observations, in this study we consider the pricing of energy derivatives assuming that the spot price is driven by CGMY and tempered stable processes of OU type. The first contribution consists in the derivation of the closed formula of the characteristic function of the transition law of these processes that is instrumental to find non-arbitrage conditions. It also gives the fundamental ingredient to calculate the price of financial derivatives with FFT-based methods. 

Based on the results of Cufaro Petroni and Sabino\mycite{cs20_3}, the second contribution is the derivation of the exact methods for the simulation of the skeleton of CGMY and bilateral tempered stable processes processes of OU type with finite variation. Particular emphasis is given to the case when such processes consists of compound Poisson processes of OU type.

The calibration and the model selection is not the focus of this study, instead we rather illustrate the theoretical applicability of our findings and the proposed simulation algorithms to the pricing of a few energy derivative contracts. As a first application, we consider the pricing of a daily strip of call options on the day-ahead spot price driven by tempered stable OU processes using the FFT technique of Carr and Madan\mycite{Carr1999OptionVU}. As mentioned, the parameters calibration is not the focus of this paper, nevertheless such a task can be easily accomplished combining the FFT pricing with an optimization to minimize the difference between the available option prices and the theoretical ones. Secondly, we consider the pricing of Asian options depending on the day-ahead spot price described by a CGMY-driven OU process via Monte Carlo simulations where we also highlight the differences between our exact simulation schemes and the standard approximation procedures. The last example consists in pricing swing options with the modified version of the Least-Squares Monte Carlo method detailed in Boogert and C. de Jong\mycite{BDJ08, BDJ10} using market models based on CGMY-driven OU process that coincide with compound Poisson processes of OU type. Finally, we show that our results are not restricted to OU processes and to the modeling of spot prices.  Indeed, in the spirit of Benth et al.\mycite{BPV19}, Latini et al.\mycite{LPV19} and Piccirilli et al.\mycite{PSV20}, they can be adapted to capture the Samuelson effect and different implied volatility profiles displayed by options futures. 

The paper is organized as follows. Section\myref{sec:preliminaries} introduces tempered stable and CGMY processes and the general results relatively to \Levy-driven OU processes. In Section\myref{sec:ou:ts:cgmy} we focus on the classical tempered stable and CGMY processes of OU type with finite variation and derive the characteristic function of their transition law. In Section\myref{sec:ou:ts:cgmy:alg} we present the algorithms for the simulation of the skeleton of the processes under study and we focus on the case of compound Poisson processes of OU type. We also present numerical experiments demonstrating their efficiency. The application of these results is illustrated in Section\myref{sec:fin:app} in the context of the pricing of energy derivative contracts, namely daily strips of call options, Asian options with European exercise, swing options written on the day-ahead spot price and futures. Finally Section\myref{sec:conclusions} concludes the paper with an overview of future inquiries and possible further applications.
								

\subsection{Notation}
Before proceeding, we introduce some notation and shortcuts that will be used throughout the paper. We write $\Gamma(\alpha, \beta)$ to denote the gamma distribution with shape parameter $\alpha>0$ and rate parameter $\beta>0$. Moreover, we write $\unif([0,1])$ to denote the uniform distribution in $[0,1]$ and $\poiss(\lambda)$ to denote the Poisson distribution with parameter $\lambda>0$.
We use the shortcuts \id\ and \sd\ for \emph{infinitely divisible} and \emph{self-decomposable} distributions, respectively. We use the shortcut \rv\ for \emph{random variable} and \iid\ for \emph{independently and identically distributed} whereas we use \chf, \lch\, \cgf\ and \pdf\ as shortcuts for \emph{characteristic function}, \emph{logarithmic characteristic}, \emph{cumulant generating function} and \emph{density function}, respectively. 

\section{Preliminaries}\label{sec:preliminaries}

Take a \Levy\ process $L(\cdot)$ of classic tempered stable  type namely, with \Levy\ measure having density 
\begin{equation}
\nu(x) = \nu_p(x) + \nu_x(x) = c_p\,\frac{e^{-\beta_p\,x}}{x^{1+\alpha_p}}\mathds{1}_{x\ge 0} + c_n\,\frac{e^{\beta_n\,x}}{|x|^{1+\alpha_n}}\mathds{1}_{x< 0}\label{eq:cts:density}
\end{equation}
where $c_p, c_n, \beta_p, \beta_n$ are all positive numbers, $\alpha_p<2$ and $\alpha_n<2$.
Hereafter we will denote with $\bcts(\alpha_p, \alpha_n, \beta_p, \beta_n, c_p\,t, c_n\,t)$ the law of $L(t)$.

Different applications of such a process can be found among other in Koponen\mycite{Koponen95}, Carr et al.\mycite{CGMY2002} Poirot and Tankov\mycite{PoirotTankov2007} and Ballotta and Kyriakou\mycite{BK14}.  In particular, the model introduced in Carr et al. is named CGMY and assumes $C=c_p=c_n$, $Y=\alpha_p=\alpha_n$, $G=\beta_n$ and $M=\beta_p$ from the names of the authors. It can also be proven that a classic tempered stable process is a time-changed Brownian motion provided $c_p=c_n$ and $\alpha_p=\alpha_n=\alpha>-1$ (see Cont and Tankov\mycite{ContTankov2004} Proposition 4.1). For sake of completeness, it is worthwhile mentioning that we are referring to classic tempered stable processes because different processes can be constructed applying an alternative tempering function rather than the exponential function used in\refeqq{eq:cts:density} (see  Rosi\'nski\mycite{ROSINSKI2007677}). An overview of such processes, named general tempered stable processes, can be found in Grabchak\mycite{Grabchak16}.

In the following we consider the subset of classic tempered stable processes with finite variation for which it holds $\alpha_p<1$ and $\alpha_n<1$ and in particular when $\alpha_p<0$ and $\alpha_n<0$ the subset consists of Poisson processes (see Cont and Tankov\mycite{ContTankov2004}). 
Due to the fact that any process of finite variation can be seen as the difference of two independent subordinators, the process $L(\cdot)$ can be written as $L(t) = L_p(t) - L_n(t)$ where $L_p(\cdot)$ and $L_n(\cdot)$ are two classic tempered stable subordinators with \Levy\ densities $\nu_p(x)$ and $\nu_n(-x)$, respectively. In the following we will dub classic tempered stable subordinators with CTS, whereas the full bilateral case will be denoted with BCTS. Moreover, we will denote the law of a CTS subordinator at time $t$ with $\cts(\alpha, \beta, c\,t)$.

Consider now an Ornstein-Uhlenbeck (OU) process $X(\cdot)$ solution of the stochastic differential equation
            \begin{eqnarray}\label{eq:genOU_sde}
              dX(t) &=&  -bX(t)dt + dL(t) \quad\qquad X(0)=X_0\quad \Pqo\qquad
              b>0
            \end{eqnarray}
namely, 
\begin{eqnarray}
X(t) &=& X_0\,e^{-bt} + Z(t) \qquad\quad Z(t)=\int_0^te^{-b\,(t-s)}dL(s). \label{eq:sol:OU}
\end{eqnarray}
\begin{eqnarray}
Z(t) &=& Z_p(t) - Z_n(t) \qquad\quad Z_d(t)=\int_0^te^{-b\,(t-s)}dL_d(s), d\in\{p, n\}. \label{eq:sub}
\end{eqnarray}
\noindent Following the convention in Barndorff-Nielsen and Shephard\mycite{BNSh01}, $X(\cdot)$ is then named OU-BCTS process or, if the above parameter constrain holds, OU-CGMY process.

There is a close relation between the concept of self-decomposability and the theory of \Levy\ driven OU processes, indeed as observed in Barndorff-Nielsen et al.\mycite{BJS1998}, the solution process\refeqq{eq:sol:OU} is stationary if and only if its
\chf $\varphi_X(u,t)$ is constant in time and steadily coincides
with the \chf\ $\overline{\varphi}_X(u)$ of the \sd\ invariant
initial distribution that turns out to be decomposable according to
\begin{equation*}
\overline{\varphi}_X(u)=\overline{\varphi}_X(u\,e^{-b\,
t})\varphi_Z(u,t)
\end{equation*}
where now, at every given $t$, $\varphi_Z(u,t)=e^{\psi_Z(u,t)}$ denotes the \id \chf\ of the \rv\ $Z(t)$ in\refeqq{eq:sol:OU} and $\psi_Z(u,t)$ its \lch. We remark that the process $Z(\cdot)$ is not a \Levy\ process but rather an additive process. 

We recall here that a law with \chf\
$\eta(u)$ is said to be \sd\ (see Sato\mycite{Sato}, Cufaro
Petroni~\cite{Cufaro08}) when for every $0<a<1$ we can find another
law with \chf\ $\chi_a(u)$ such that
\begin{equation}\label{aremchf}
    \eta(u)=\eta(au)\chi_a(u).
\end{equation}
Of course a \rv $X$ with \chf\ $\eta(u)$ is also
said to be \sd\ when its law is \sd, and looking at the definitions
this means that for every $0<a<1$ we can always find two
\emph{independent} \rv's -- a $Y$ with the same law of $X$, and a
$Z_a$ with \chf\ $\chi_a(u)$ -- such that in distribution
\begin{equation}\label{sdec-rv}
    X\eqd aY+Z_a
\end{equation}
Hereafter the \rv\ $Z_a$ will be called the \emph{\arem} of $X$ and
in general has an \id\ (see Sato\mycite{Sato})

This last statement apparently means that the law of $Z(t)$ in the
solution\refeqq{eq:sol:OU} coincides with that of the \arem\ of the
\sd, stationary law $\overline{\varphi}_X$ provided that $a=e^{-b\,
t}$. It is easy indeed to see
from\refeqq{eq:sol:OU} that the \chf\ of the time homogeneous
transition law with a degenerate initial condition $X(0)=x_0,\;\Pqo$
is
\begin{equation}\label{transchf}
    \varphi_X(u,t|x_0)=e^{\,ix_0ue^{-bt}}\varphi_Z(u,t)=\frac{\overline{\varphi}_X(u)\,e^{\,ix_0ue^{-bt}}}{\overline{\varphi}_X(u\,e^{-b\, t})}
\end{equation}
moreover we have
\begin{equation}
    \psi_Z(u, t) = \psi_{Z_p}(u) + \psi_{Z_n}(-u) \label{eq:ou:cts:lch}
\end{equation}
and
\begin{equation}
    \psi_X(u,t|x_0)=iux_0e^{-bt} + \psi_Z(u,t) = iux_0e^{-bt} + \psi_{Z_p}(u, t) + \psi_{Z_n}(-u, t).
\label{eq:cumulant:function}
\end{equation}
where $\psi_{Z_d}(u,t)=\ln\varphi_{Z_d}(u,t), d\in\{p, n\}$ is the \lch\ of $Z_d(t), d\in\{p, n\}$.
Moreover, the transition \lch\ of a OU process can also
be written in terms of the corresponding $\psi_L(u)$  in
the form
\begin{equation}
    \psi_X(u,t|x_0)=iux_0e^{-bt} + \psi_Z(u,t) = iux_0e^{-b\, t} + \int_0^t\psi_L\left(ue^{-b\,s}\right)ds.
\label{eq:lch:ou}
\end{equation}
Finally in virtue of the results of Cufaro Petroni and Sabino\mycite{cs20_3}, one can relate the \Levy\ density $\nu_Z(x,t)$ of $Z(t)$ to that of the BDLP $L(\cdot)$ at $t=1$ denoted with $\nu_L(x)$ 
\begin{eqnarray}
    \nu_Z(x,t) &=& \frac{1}{b\,|x|}
              \left\{
                   \begin{array}{ll}
                   \int_{\,^x\!/_a}^x \nu_L(y)dy   & \;\hbox{$x<0$} \\
                                 \\
                   \int_x^{\,^x\!/_a} \nu_L(y)dy   & \;\hbox{$x>0$}
                   \end{array}
             \right. \quad a = e^{-bt}.
    \label{eq:levy:measures:ou++_Z}
\end{eqnarray}


\section{OU-BCTS and OU-CGMY processes}\label{sec:ou:ts:cgmy}
In this section we study OU-BCTS and OU-CGMY processes with finite variation and distinguish the case where the BDLP is of infinite activity, $0<\alpha_p<1$ and $0<\alpha_n<1$, to that of finite activity namely when $Z(\cdot)$ is a compound Poisson process. We do not discuss the setting $\alpha_p=\alpha_n$ because it is already covered in Sabino\mycite{Sabino20b} and corresponds Variance Gamma driven OU process, therefore of infinite activity and finite variation. 

\subsection{Infinite activity}\label{subsec:ou:ts:infinite}
Apparently, the study of the transition law of a OU-BCTS process $X(\cdot)$ coincides with the study of the process $Z(\cdot)$ and in particular of the processes $Z_p(t)$ and $Z_n(t)$ defined in\refeqq{eq:sub}. One of course can rely on these last two processes to build OU-CTS processes. 

Cufaro Petroni and Sabino\mycite{cs20_3} and Qu et al.\mycite{QDZ20} designed an exact decomposition of the transition law of OU-CTS processes as the convolution of two independent \rv's plus a degenerate term. For simplicity, we report this result here below in addition because, such a OU process is driven by a CTS subordinator, we consider only $Z_p(\cdot)$.

\begin{prop}\label{prop:ou:ts}
For $0< \alpha_p < 1 $, and at every $t>0$, the pathwise
solution of an OU-CTS
equation\refeqq{eq:sol:OU} with $X(0)=X_0,\;\Pqo$ is in
distribution the sum of three independent \rv's
\begin{equation}
X(t)=aX_0 + Z_p(t)\, \eqd\, aX_0+X_1 + X_2 \qquad\quad a=e^{-b\,
t}\label{eq:prop:outs_1}
\end{equation}
where $X_1$ is distributed according to the law $\cts\!\left(\alpha_p,
\frac{\beta_p}{a}, c_p\,\frac{1 - a^\alpha_p}{\alpha_p\, b}\right)$, whereas
\begin{equation*}
X_2=\sum_{k=1}^{N_a} J_k
\end{equation*}
is a compound Poisson \rv\ where $N_a$ is an independent Poisson
\rv\ with parameter
\begin{equation}
 \Lambda_a =
 \frac{c_p\,\beta_p^{\alpha_p}\Gamma(1-\alpha_p)}{b\,\alpha_p^2 a^{\alpha_p}}\,\left(1-a^{\alpha_p}+a^{\alpha_p}\log a^{\alpha_p}\right) \label{eq:outs:intensity}
\end{equation}
and $J_k, k>0$ are \iid\ \rv's with density
\begin{equation}
 f_J(x)=\frac{\alpha_p\, a^{\alpha_p}}{1-a^\alpha_p+a^{\alpha_p}\log a^{\alpha_p}} \int_1^{\frac{1}{a}}\frac{x^{-\alpha_p} \left(\beta_p\,
  v\right)^{1-\alpha_p}e^{-\beta_p v\, x}}{\Gamma(1-\alpha_p)}\, \frac{v^{\alpha_p}-1}{v}\, dv \label{eq:outs:jumps}
\end{equation}
namely, a mixture of a gamma law and a distribution with density
\begin{equation}
  f_V(v) = \frac{\alpha_p\, a^\alpha_p}{1-a^{\alpha_p}+a^{\alpha_p}\log
    a^{\alpha_p}}\, \frac{v^{\alpha_p}-1}{v}\qquad\quad 1\le v\le\,^1/_a.\label{eq:ou:cts:mixture:rate}
\end{equation}
\end{prop}
The extension to the bilateral OU-BCTS process is straightforward, for instance the simulation algorithms consist of repeating the procedure for an OU-CTS process two times. 

The main contribution of this subsection is the derivation of the \lch\ and hence the \chf\ and the moment generating function of $Z(t)$ that will be instrumental to find the risk neutral conditions for market models based on OU-BCTS processes and to the pricing of derivative contracts using FFT methods.

\begin{prop}\label{prop:chf:sub}
The \lch\ $\psi_{Z_p}(u, t)$, $u\in\R$ with $0<\alpha_p<1$ can be represented as:
\begin{equation}
\psi_{Z_p}(u, t) = -\frac{c_p\,\beta_p^{\alpha_p}\,\Gamma(1-\alpha_p)}{\alpha_p\,b}\left[I\left(u, \alpha_p, \beta_p, \frac{\beta_p}{a}\right) + \log a\right],\quad a=e^{-b\,t}
\label{eq:lch:ou:cts}
\end{equation}
with
\begin{eqnarray}
I\left(u, \alpha, \beta_1, \beta_2\right) &=& \int_{\beta_1}^{\beta_2} z^{-1-\alpha}(z - iu)^{\alpha}dz=\nonumber\\
&=&-\frac{1}{\alpha}\left[\left(\frac{u}{i\,\beta_2}\right)^{\alpha}\, _2F_1\left(-\alpha, -\alpha, 1-\alpha, -\frac{i\,\beta_2}{u}\right) \right.- \nonumber\\
&=&\left.\left(\frac{u}{i\,\beta_1}\right)^{\alpha}\,_2F_1\left(-\alpha, -\alpha, 1-\alpha, -\frac{i\,\beta_1}{u}\right)  \right] 
\label{eq:J_W}
\end{eqnarray}
\noindent where $_2F_1(a, b, c, x)$ is the hypergeometric function, $0<\alpha<1$, $\beta_1>0$ and $\beta_2>0$.
Finally, taking $\alpha=1/2$ the hypergeometric function in\refeqq{eq:J_W} can be expressed in terms of elementary functions as
\begin{equation*}
_2F_1\left(-\frac{1}{2}, -\frac{1}{2}, \frac{1}{2}, -x\right) =\sqrt{x + 1} - \sqrt{x}\sinh^{-1}\left(\sqrt{x}\right) 
\end{equation*}
\noindent where $\sinh^{-1}x = \log\left(\sqrt{x^2 + 1} + x\right)$.

\end{prop}
\begin{proof}
From\refeqq{eq:levy:measures:ou++_Z} it also results
\begin{equation*}
 \nu_{Z_p}(u,t) = \frac{c_p}{b\,x}\int_{x}^{\frac{x}{a}}\frac{e^{-\beta_p y}}{y^{\alpha_p + 1}}\,dy
             \;=\; \frac{c_p}{b}\int_{1}^{\frac{1}{a}}\frac{e^{-\beta_p wx}}{x^{\alpha_p + 1}\, w^{\alpha_p + 1}}\,dw 						
\end{equation*}
therefore because of the \LKh\ theorem 
\begin{equation*}
 \psi_{Z_p}(u,t) = \frac{1}{b}\int_{1}^{\frac{1}{a}}w^{-\alpha_p - 1}dw\int_0^{\infty}c_p\,\left (e^{i\,u\,x} - 1\right) \frac{e^{-\beta_p wx}}{x^{\alpha_p + 1}}\,dx 						
\end{equation*}
The second integral is the \lch\ of a $\cts(\alpha_p, \beta_p, c_p)$ law with $0<\alpha_p<1$ therefore from Lemma 2.5 of K\"{u}chler and Tappe\mycite{KT2013} we have
\begin{eqnarray*}
 \psi_{Z_p}(u,t) &=& \frac{c_p\,\Gamma(-\alpha_p)}{b}\int_{1}^{\frac{1}{a}}\frac{\left(\beta_p\,w - i\,u\right)^{\alpha_p} - (\beta_p\,w)^{\alpha_p}}{w^{\alpha_p + 1}}dw\\
&=& \frac{c_p\,\Gamma(-\alpha_p)}{b}\left[\int_{1}^{\frac{1}{a}}\frac{\left(\beta_p\,w - i\,u\right)^{\alpha_p}}{w^{\alpha_p + 1}}dw - \beta_p^{\alpha_p}\int_a^{\frac{1}{a}}\frac{dw}{w}\right]\\
&=& \frac{c_p\,\beta_p^{\alpha}\,\Gamma(-\alpha_p)}{b}\left[\int_{\beta_p}^{\frac{\beta_p}{a}} z^{-1-\alpha_p}(z - iu)^{\alpha_p}dz + \log a\right] \\
\end{eqnarray*}
where of course $-\alpha_p\Gamma(-\alpha_p)=\Gamma(1-\alpha_p)$ and in last step we have used the change of variables $	\beta_p\,w = z$. 
In order to conclude the proof, we first observe that under the special case $\gamma = \alpha + 1$ the derivative of the hypergeometric function is
\begin{equation*}
\frac{d}{d\,x}\,_2F_1\left(\alpha, \beta, \alpha+1, x\right) = \frac{d}{d\,x}\,_2F_1(\beta, \alpha, \alpha+1, x) = 
-\frac{\alpha \left( (1-x)^{\beta} -\, _2F_1(\alpha, \beta, \alpha+1, x\right))}{z}
\end{equation*}
then with some algebra we get
\begin{equation*}
\frac{d}{d\,x}\left(z^{-\alpha}\,_2F_1\left(-\alpha, -\alpha, 1-\alpha, -\frac{i\,z}{u}\right)\right) = 
\alpha\left(\frac{i}{u}\right)^{\alpha}z^{\alpha + 1}\,\left(z - i\,u\right),
\end{equation*}
therefore we can write the integral
\begin{eqnarray*}
I\left(u, \alpha, \beta_p, \frac{\beta_p}{a}\right)&=&-\frac{1}{\alpha}\left[\left(\frac{u}{i\,\beta_2}\right)^{\alpha}\, _2F_1\left(-\alpha, -\alpha, 1-\alpha, -\frac{i\,\beta_2}{u}\right) \right.- 
\nonumber\\
&=&\left.\left(\frac{u}{i\,\beta_1}\right)^{\alpha}\,_2F_1\left(-\alpha, -\alpha, 1-\alpha, -\frac{i\,\beta_1}{u}\right)  \right] 
\end{eqnarray*}
as claimed.
\end{proof}
\begin{remark}
Several transformation and recursion formulas are applicable to the hypergeometric functions. Using 9.131.1 in Gradshteyn and Rizhik\mycite{gradshteyn2007} we have
\begin{equation*}
_2F_1(-\alpha, -\alpha, 1-\alpha, x) = (1-x)^{\alpha + 1}\,_2F_1(1, 1, 1-\alpha, x)
\end{equation*} 
and accordingly\refeqq{eq:J_W} becomes
\begin{eqnarray}
I\left(u, \alpha, \beta_1, \beta_2\right) 
&=&-\frac{i}{\alpha\,u}\left[\beta_2 ^{-\alpha}\left(\beta_2 - i\,u\right)^{\alpha+1} \, _2F_1\left(1, 1, 1-\alpha, -\frac{i\,\beta_2}{u}\right) \right.- \nonumber\\
&&\left.\beta_1 ^{-\alpha}\left(\beta_1 - i\,u\right)^{\alpha+1}\, _2F_1\left(1, 1, 1-\alpha, -\frac{i\,\beta_1}{u}\right) \right]
\label{eq:J_W_2}
\end{eqnarray}

\end{remark}
\begin{cor}\label{corr:mgf:sub}
The \cgf\ $m_{Z_p}(s,t)= \ln\EXP{e^{s\,Z_p(t)}}$ of $Z_p(t)$ with $0<\alpha_p<1$ exists for $s<\beta_p$ and is:
\begin{equation}\label{eq:mgf:ou:cts}
 m_{Z_p}(s,t) = -\frac{c\,\beta_p^{\alpha_p}\,\Gamma(1-\alpha_p)}{\alpha_p\,b}\left[\tilde{I}\left(s, \alpha_p, \beta_p, \frac{\beta_p}{a}\right) + \log a\right] \quad s<\beta_p
\end{equation}
where
\begin{eqnarray}
\tilde{I}\left(s, \alpha, \beta_1, \beta_2\right) &=& \int_{\beta_1}^{\beta_2} z^{-1-\alpha}(z - s)^{\alpha}dz=\nonumber\\
&&
\frac{1}{\alpha\,s}\left[\beta_2 ^{-\alpha}\left(\beta_2 - s\right)^{\alpha+1} \, _2F_1\left(1, 1, 1-\alpha, \frac{\beta_2}{s}\right) \right.- \nonumber\\
&&\left.\beta_1 ^{-\alpha}\left(\beta_1 - s\right)^{\alpha+1}\, _2F_1\left(1, 1, 1-\alpha, -\frac{\beta_1}{s}\right) \right]\label{eq:I_tilde}
\end{eqnarray}
\end{cor}
Note that setting $\psi_{Z_p}(-i\,s, t) = m_{Z_p}(s, t)$ in\refeqq{eq:J_W} one may claim that the \cgf\ assumes complex values which is obviously wrong and it explains why we have preferred to rely on\refeqq{eq:J_W_2} to write\refeqq{eq:I_tilde}.

In virtue of\refeqq{eq:ou:cts:lch}, Proposition\myref{prop:chf:sub} and Corollary\myref{corr:mgf:sub} can be easily extended to cope with $Z(t)$ defined for OU-BCTS processes. 
\begin{cor}\label{corr:ou:bcts:lch}
For a OU-BCTS process, the \lch\ $\psi_{Z}(u, t)$, $u\in\R$ can be represented as:
\begin{eqnarray}
\psi_Z(u, t) 
&=& -\frac{c_p\,\beta_p^{\alpha_p}\,\Gamma(1-\alpha_p)}{\alpha_p\,b}\left[I\left(u, \alpha_p, \beta_p,\frac{\beta_p}{a}\right) + \log a\right] - \\
&& \frac{c\,\beta_n^{\alpha_n}\,\Gamma(1-\alpha_n)}{\alpha_n\,b}\left[I\left(-u, \alpha_n, \beta_n, \frac{\beta_n}{a}\right) + \log a\right]\label{eq:chf:ou:bcts}
\end{eqnarray}
\end{cor}
Accordingly,
\begin{cor}
The \cgf\ at time $t$ exists for $-\beta_n<s<\beta_p$ and is:
\begin{eqnarray}
 m_Z(s,t) &=& -\frac{c\,\beta_p^{\alpha_p}\,\Gamma(1-\alpha_p)}{\alpha_p\,b}\left[\tilde{I}\left(s, \alpha_p, \beta_p, \frac{\beta_p}{a}\right) + \log a\right] - \nonumber \\
&& \frac{c\,\beta_n^{\alpha_n}\,\Gamma(1-\alpha_n)}{\alpha_n\,b}\left[\tilde{I}\left(-s, \alpha_n, \beta_n,\frac{\beta_n}{a}\right)  + \log a\right]
\label{eq:mgf:ou:bcts}
\end{eqnarray}
\end{cor}

\begin{remark}
In contrast to Proposition\myref{prop:ou:ts} that is valid under the condition $0<\alpha_p<1$, we will show in the next subsection that Proposition\myref{prop:chf:sub} and consequently all corollaries are also valid for $\alpha_p<0$ and $\alpha_n<0$ as well.   
\end{remark}



\subsection{Finite activity}\label{subsec:fin:act}
When $\alpha_p <0$ the BDLP of a CTS process turns out to be a compound Poisson process, indeed the integral of the \Levy\ density is convergent. In more detail we have:
\begin{equation}
\int_0^{\infty}\nu_{L_p}(x)dx = c_p\int_0^{\infty}x^{-\alpha_p-1}e^{-\beta_p\,x}dx = 
c_p\,\Gamma(-\alpha_p)\,\beta_p^{-\alpha_p} = \lambda_p \label{eq:ts:finite:lambda}
\end{equation}
where now $-\alpha$ is positive. It results then that
\begin{equation}
L_p(t) = \sum_{k=0}^{N_p(t)}J_k, \quad J_0=0, \Pqo,
\label{eq:ts:finite}
\end{equation}
where $N_p(t)$ is a Poisson process with the intensity $\lambda_p$ defined in\refeqq{eq:ts:finite:lambda} and jumps sizes $J_k$ independent on $N_p(t)$ distributed according to a gamma law with shape parameter $-\alpha_p$ and rate parameter $\beta_p$ indeed the \pdf\ of each copy of $J_k$ is 
\begin{equation*}
f_J(x) = \frac{\nu_L(x)}{\lambda} = \frac{\beta_p^{-\alpha_p}\,x^{-\alpha_p - 1\, e^{-\beta_p\,x}}}{\Gamma(-\alpha_p)}, \quad x>0.
\end{equation*}
This last representation is consistent with J{\o}rgensen\mycite{Jorgensen97} that observed that a compound Poisson process with gamma distributed jumps follows a Tweedie distribution that is actually a CTS law.

\begin{prop}\label{prop:finite:ts}
For $\alpha_p < 0 $, and at every $t>0$, the pathwise
solution of an OU-CTS
equation\refeqq{eq:sol:OU} with $X(0)=X_0,\;\Pqo$ is in
distribution the sum of two independent \rv's
\begin{equation}
X(t)=aX_0 + Z_p(t)\, \eqd\, aX_0+X_1 \qquad\quad a=e^{-b\,
t}\label{eq:prop:outs_2}
\end{equation}
where $X_1$ can be written as $\sum_{k=0}^{N_p(t)}\tilde{J}_k$. $N_p(t)$ is a Poisson process with intensity $\lambda_p$ given by equation\refeqq{eq:ts:finite:lambda} and $\tilde{J}_k, k>0$ are \iid\ jumps distributed according to a mixture of gamma law and a uniform distribution with \pdf
\begin{equation}
\int_0^1 \frac{\left(\beta_p\,e^{b\,v\,t}\right)^{-\alpha_p}\,x^{-\alpha_p - 1}e^{-\beta_p\,e^{b\,v\,t}\,x}}{\Gamma(-\alpha_p)}dv
\label{eq:gamma:mix}
\end{equation}
or equivalently $\tilde{J_k}\sim\Gamma(-\alpha_p, \beta_p\, e^{b\,U_k t}), k>0$ where $U_k\sim\unif(0,1)$
\end{prop}
\begin{proof}
According to the definition of the OU-CTS process for $\alpha_p<0$ and the representation\refeqq{eq:ts:finite} we can write
\begin{equation*}
X(t) = a\, X_0 + \sum_{k=0}^{N_p(t)}J_k\,e^{-b(t-\tau_k)}
\end{equation*}
where $\tau_k$ are the jump times of the Poisson process $N_p(t)$ with intensity $\lambda_p$.
On the other hand, as observed by
Lawrance\mycite{L80} in the context of Poisson point processes, 
for every $t>0$ we have
\begin{equation*}
\sum_{k=0}^{N_p(t)}J_k e^{-b(t-\tau_k)}\;\eqd\; \sum_{k=0}^{N_p(t)}J_k
e^{-b\,t\,U_k}, \quad U_0=0,  \Pqo \label{eq:law}
\end{equation*}
irrespective of the law of $J_k$,
where $U_k\sim\unif([0,1]), k>0$ form a sequence of \iid\ uniformly  distributed \rv's in $[0, 1]$, also independent on $J_k$. Knowing that for any gamma distributed random variable $Y\sim\Gamma(\alpha, \beta)$, and $A>0$, $A\, Y\sim\Gamma\left(\alpha, \frac{\beta}{A}\right)$, it results that $\tilde{J}_k\sim\Gamma(-\alpha_p, \beta_p\,e^{b\, U_k}), k>0$ that concludes the proof.
\end{proof}
Because with $\alpha_p<0$, $Z_p(t)$ is distributed according to a compound Poisson \rv\  its \lch\ is
\begin{eqnarray*}
\psi_{Z_p}(u, t) &=& \lambda_p\,t\left(\int_0^1\left(\frac{\beta_p\,e^{b\,v\,t}}{\beta_p\,e^{b\,v\,t} -i\,u}\right)^{-\alpha_p}dv -1\right)=\\
&=& 
\frac{\lambda_p}{b}
\left(\int_{\beta_p}^{\beta_p\,e^{b\,t}}
z^{-\alpha_p-1}(z-i\,u)^{\alpha_p}  dz -b\,t\right)
\end{eqnarray*}
where in the last step we used the change of variables $z=\beta_p\,e^{b\,v\,t}$. Replacing $a=e^{-b\,t}$ and $\lambda_p$ in\refeqq{eq:ts:finite:lambda} we get the same representation of $\psi_{Z_p}(u,t)$ as that in Proposition\myref{prop:chf:sub}; of course all corollaries of subsection\myref{subsec:ou:ts:infinite}  follow accordingly.

\section{Simulation Algorithms}\label{sec:ou:ts:cgmy:alg}
The sequential generation of the skeleton of an OU-BCTS or an OU-CGMY process on
a time grid $t_0, t_1, \dots, t_I$ simply consists in implementing the following recursive procedure with
initial condition $X(t_0)=x_0$ taking $a_i=e^{-b(t_{i}-t_{i-1})},\; i=1,\dots, I$:
\begin{equation}
X(t_{i}) = a_i X(t_{i-1}) + Z_{a_i}, \quad\qquad i=1,\dots, I.
\label{eq:ou:recursion}
\end{equation}
Cufaro Petroni and Sabino\mycite{cs20_3} have already discussed algorithms tailored for OU-CTS processes of infinite activity, the extension to bilateral OU-BCTS or OU-CGMY processes is straightforward. 

In this section we illustrate the simulation procedure when these processes are of finite activity that to our knowledge has not been investigated so far. To this end, the simulation steps to generate the skeleton of a OU-CTS process with parameters $b, \alpha_p, \beta_p, c_p$ is summarized in Algorithm\myref{alg:ou:cts}

\begin{algorithm}
\caption{ }\label{alg:ou:cts}
\begin{algorithmic}[1]
        \State $X_0\gets x$
        \For{ $i=1, \dots, I$}
        \State $\Delta t_i=t_i - t_{i-1}$, $a \gets e^{-b\Delta t_i}$
        \State $n\gets N\sim\poiss(\lambda_p\,\Delta t_i)$,
        \Comment{Generate an independent Poisson \rv\ with $\lambda$ in\refeqq{eq:ts:finite:lambda}}
        \State $u_m\gets U_m\sim\unif(0,1), m=1, \dots, n$
        \Comment{Generate $n$ \iid\ uniform \rv's}
        \State $\tilde{\beta}_m \gets \beta_p\, e^{b\,u_m}, m=1, \dots, n$
        \State $\tilde{j_m}\gets \tilde{J_m}\sim\gam(1-\alpha,\, \tilde{\beta}_m), m=1, \dots, n$
        \Comment{Generate $n$ independent gamma \rv's with scale $-\alpha_p$ and random rates $\tilde{\beta}_m$}
        \State $x_1\gets \sum_{m=1}^n\tilde{j_m}$
        \State $X(t_i)\gets a\,X(t_{i-1}) + x_1$.
        \EndFor
        \end{algorithmic}
\end{algorithm}
We remark that when $-\alpha_p=1$ the BDLP of the OU-CTS process is a compound
Poisson process with exponential jumps that corresponds to a OU process with a gamma stationary law. For this configuration it is preferable to use the faster and more efficient algorithm detailed in Sabino and Cufaro Petroni\mycite{cs20_2, cs20_1}.

Finally, the procedure the generate the skeleton of OU-BCTS and OU-CGMY processes simply entails to repeat steps $5$ to $8$ two times and add their outcome to step $9$. 


\subsection{Numerical Experiments}\label{sub:sec:ou:ts:cgmy:num}
In this section, we will assess the performance and the
effectiveness of the algorithms for the simulation of OU-BCTS process. All the simulation experiments in the present paper
have been conducted using \emph{Python} with a $64$-bit Intel Core
i5-6300U CPU, 8GB. The performance of the algorithms is ranked in
terms of the percentage error relatively to the first four cumulants
denoted \emph{err} \% and defined as
\begin{equation*}
   \text{err \%} = \frac{\text{true value} - \text{estimated
value}}{\text{true value}}
\end{equation*}
Taking advantage of\refeqq{eq:lch:ou} one can calculate the cumulants 
$c_{X,k}(x_0,t),\;k= 1,2,\ldots$ of $X(t)$ for $X_0=x_0$ from the
cumulants $c_{L,k}$ of the BCTS law according to
\begin{eqnarray}
c_{X,1}(x_0,t) &=&\EXP{X(t)|X_0=x_0}\;=\; x_0e^{-b\, t} + \frac{c_{L,1}}{b}\left(1-e^{-b\, t}\right),\qquad k=1\label{eq:cumulants:ou1}\\
c_{X,k}(x_0,t) &=& \frac{c_{L,k}}{k\,b}\left(1-e^{-k\,b\,
t}\right), \qquad\qquad\qquad\qquad\qquad\qquad
k=2,3,\ldots\label{eq:cumulants:ou2}
\end{eqnarray}
where
\begin{equation}\label{eq:cts:cumulants}
c_{L,k} =\int_{-\infty}^{+\infty}x^k\nu_L(x)\,dx=
c_p\,\beta_p^{\alpha_p-k}\Gamma(k-\alpha_p) + (-1)^k c_n\,\beta_n^{\alpha_n-k}\Gamma(k-\alpha_n)
\end{equation}
In our numerical experiments we consider a OU-CTS process with
parameters $\left(b, \beta_p, c_p \right) = \left(0.5, 1.5, 0.3\right)$ and a OU-CGMY process with $\left(b, C, G, M\right)=(0.5, 0.3, 0.5, 1.5)$
with $\alpha_p=Y\in\{-0.5, -1.5, -2.5, -3.5\}$. 

The Tables\myref{tab:ou:cts:1:12} and\myref{tab:ou:bcts:1:2}
compare then the true values of the first four cumulants
$c_{X,k}(0,\Delta t)$ with their corresponding estimates from $10^6$
Monte Carlo (MC) simulations respectively for the OU-CTS process with $\Delta t=1/12$ and for the OU-CGMY process with $\Delta
t=1/2$, each of the two with the aforementioned parameters. We can conclude therefrom that the proposed
Algorithm\myref{alg:ou:cts} and its adaptation to the bilateral case produce unbiased cumulants that are very
close to their theoretical values. For the sake of brevity, we do
not report the additional results obtained with different parameter
settings that anyhow bring us to the same findings. Overall, from
the numerical results reported in this section, it is evident that
the Algorithm\myref{alg:ou:cts} proposed above can achieve a very
high level of accuracy as well as a conspicuous efficiency.

\begin{table}[ht!]
    \centering\scriptsize
        \resizebox{\textwidth}{!}{
        \begin{tabular}{*{13}{|c|rrr|rrr|rrr|rrr}}
                    \hline
&       \multicolumn{3}{c|}{$c_{X, 1}(0,\Delta t)$} & \multicolumn{3}{c|}{$c_{X, 2}(0,\Delta t)$} & \multicolumn{3}{c|}{$c_{X, 3}(0,\Delta t)$} & \multicolumn{3}{c|}{$c_{X, 4}(0,\Delta t)$} \\
                        \hline
                    \hline
                    $\alpha_p$ & true & MC & err \% & true & MC & err \% & true & MC & err \% & true & MC & err \% \\
                    \hline
$-0.5$ & $1.181$ & $1.199$ & $-1.5$ & $1.157$ & $1.182$ & $-2.2$ & $1.889$ & $1.910$ & $-1.1$ & $4.320$ & $4.504$ & $-4.3$\\
$-1.5$ & $1.181$ & $1.205$ & $-2.0$ & $1.929$ & $2.020$ & $-4.7$ & $4.409$ & $4.616$ & $-4.7$ & $12.960$ & $13.422$ & $-3.6$\\
$-2.5$ & $1.969$ & $1.929$ & $2.0$ & $4.500$ & $4.378$ & $2.7$ & $13.226$ & $12.636$ & $4.5$ & $47.520$ & $45.830$ & $3.6$\\
$-3.5$ & $4.594$ & $4.539$ & $1.2$ & $13.500$ & $13.335$ & $1.2$ & $48.496$ & $48.099$ & $0.8$ & $205.921$ & $206.598$ & $-0.3$\\
                    \hline
        \end{tabular}
        }
    \scriptsize
    \caption{\footnotesize{Comparing the first four true cumulants
    with their corresponding MC-estimated values (multiplied by $100$)
obtained with $10^6$ simulations and $\Delta t=1/12$, $(b, \beta_p, c_p)=(0.5, 1.5, 0.3)$.}}\label{tab:ou:cts:1:12}
\end{table}

\begin{table}[ht!]
    \centering\scriptsize
        \resizebox{\textwidth}{!}{
        \begin{tabular}{*{13}{|c|rrr|rrr|rrr|rrr}}
                    \hline
&       \multicolumn{3}{c|}{$c_{X, 1}(0,\Delta t)$} & \multicolumn{3}{c|}{$c_{X, 2}(0,\Delta t)$} & \multicolumn{3}{c|}{$c_{X, 3}(0,\Delta t)$} & \multicolumn{3}{c|}{$c_{X, 4}(0,\Delta t)$} \\
                        \hline
                    \hline
                    $Y$ & true & MC & err \% & true & MC & err \% & true & MC & err \% & true & MC & err \% \\
                    \hline
$-0.5$ & $-0.269$ & $-0.270$ & $-0.4$ & $0.945$ & $0.953$ & $-0.9$ & $-3.883$ & $-3.936$ & $-1.4$ & $25.13$ & $25.59$ & $-1.8$\\
$-1.5$ & $-0.934$ & $-0.935$ & $-0.1$ & $4.533$ & $4.526$ & $0.2$ & $-27.58$ & $-27.41$ & $0.6$ & $225.1$ & $221.1$ & $1.8$\\
$-2.5$ & $-4.883$ & $-4.890$ & $-0.1$ & $31.29$ & $31.35$ & $-0.2$ & $-249.4$ & $-249.3$ & $0.0$ & $2473$ & $2463$ & $0.4$\\
$-3.5$ & $-34.68$ & $-34.65$ & $0.1$ & $280.3$ & $280.2$ & $0.0$ & $-2747$ & $-2746$ & $0.1$ & $32127$ & $31636$ & $1.5$\\
                    \hline
        \end{tabular}
        }
    \scriptsize
    \caption{\footnotesize{Comparing the first four true cumulants
    with their corresponding MC-estimated values
obtained with $10^6$ simulations and $\Delta t=1/2$, $(b, C, G, M)=(0.5, 0.3, 0.5, 1.5)$.}}\label{tab:ou:bcts:1:2}
\end{table}


\section{Financial Applications}\label{sec:fin:app}

In the following subsections we illustrate the application of the results shown in Section\myref{sec:ou:ts:cgmy} and of the simulation algorithm of Section\myref{sec:ou:ts:cgmy:alg}  to the pricing of derivative contracts in energy markets using models driven by OU-BCTS and OU-CGMY processes. Energy markets and wider commodities markets exhibit mean-reversion, seasonality and spikes, this last feature is particularly difficult to  capture with a pure Gaussian framework and motivates the use of \Levy\ process. To this end, the literature is very rich of alternatives, 
for instance Cartea and Figueroa\mycite{CarteaFigueroa} assumes that the evolution of the spot prices follows a jump-diffusion OU process, 
whereas Meyer-Brandis and Tankov\mycite{KT06} investigate the use of generalized OU processes. 

Our model is similar to that of Benth et al.\mycite{BKM07} and Benth and Benth\mycite{BenthBenth04} where instead of NIG processes, we consider BCTS or CGMY processes as BDLP's. Our main goal, is to give the basis for the theoretical pricing and to provide an efficient and exact simulation procedure rather to focus on the parameter calibration and the model selection. Indeed, such dynamics based on OU-BCTS and OU-CGYM processes can also find application in other financial contexts. 

Our financial applications consider the pricing of a strip of call options with a FFT-based approach, the evaluation of a forward start Asian option with MC simulations and finally the pricing of a swing option using a modified version of version of the Least-Squares Monte Carlo (LSMC), introduced in Longstaff-Schwartz\mycite{LSW01}, detailed in Boogert and C. de Jong\mycite{BDJ08, BDJ10}.

We assume that the spot price is driven by the following one-factor process
	\begin{equation}\label{eq:market}
		S(t) = F(0,t)\, e^{h(t) + X(t)}
	\end{equation}
	where $h(t)$ is a deterministic function, $F(0,t)$ is the forward curve derived from quoted products and $X(t)$ is a OU-BCTS process. 
This market can easily be turned into a multi-factor one, for instance adding a second CTS process obtaining a tempered stable version of the two factor Gaussian model of Schwartz and Smith\mycite{SchwSchm00}. We nevertheless focus on the model\refeqq{eq:market} to better highlight the results obtained for the OU-BCTS and OU-CGMY processes.
	
Using  Lemma 3.1 in Hambly
et al.\mycite{HHM11},  the risk-neutral conditions are met when the deterministic function $h(t)$ is consistent with forward curve such that
\begin{equation}\label{eq:rn:spot}
	h(t) = -m_X(1, t)
\end{equation}
where $m_X(s, t)$ is the \cgf\ $m_X(s, t) = s\,e^{-bt} + m_Z(s, t)$ and $m_Z(s, t)$ is given by\refeqq{eq:mgf:ou:bcts}, therefore 
\begin{eqnarray}\label{eq:rn:ou:bcts}
	h(t) &=& \frac{c\,\beta_p^{\alpha_p}\,\Gamma(1-\alpha_p)}{\alpha_p\,b}\left[\tilde{I}\left(1, \alpha_p, \beta_p,\frac{\beta_p}{a}\right) + \log a\right] + \nonumber \\
&& \frac{c\,\beta_n^{\alpha_n}\,\Gamma(1-\alpha_n)}{\alpha_n\,b}\left[\tilde{I}\left(-1, \alpha_n, \beta_n, \frac{\beta_n}{a}\right) + \log a\right]
\end{eqnarray}
with $\beta_p>1$ and $\beta_n>0$.
When $\alpha_p=\alpha_n=1/2$, the integrals can be written in terms of the logarithmic function as follows
\begin{eqnarray}
\tilde{I}\left(1, \frac{1}{2}, \beta_1, \beta_2\right)&=&  \int_{\beta_1}^{\beta_2} z^{-\frac{3}{2}}(z - 1)^{\frac{1}{2}}dz = \nonumber\\
&=&2\,\beta_2^{\alpha}\,\left(\log \left(\sqrt{\beta_2} + \sqrt{\beta_2 - 1}\right) -  \sqrt{\frac{\beta_2 - 1}{\beta_2}}\right) - \nonumber\\
&&2\,\beta_1^{\alpha}\,\left(\log \left(\sqrt{\beta_1} + \sqrt{\beta_1 - 1}\right) -  \sqrt{\frac{\beta_1 - 1}{\beta_1}}\right),
\label{eq:rn:integrals_p_0.5}\\
\tilde{I}\left(-1, \frac{1}{2}, \beta_1, \beta_2\right) &=&\int_{\beta_1}^{\beta_2} z^{-\frac{3}{2}}(z + 1)^{\frac{1}{2}}dz = \nonumber\\
&=&2\,\beta_2^{\alpha}
\left(\log \left(\sqrt{\beta_2} + \sqrt{\beta_2 + 1}\right) -  \sqrt{\frac{\beta_2+1}{\beta_2}}\right)-\nonumber\\
&&2\,\beta_1^{\alpha}
\left(\log \left(\sqrt{\beta_1} + \sqrt{\beta_1 + 1}\right) -  \sqrt{\frac{\beta_1+1}{\beta_1}}\right).
\label{eq:rn:integrals_n_0.5}
\end{eqnarray}

\subsection{Call Options}\label{subsec:fin:app:call}
We consider a daily strip of $M$ call options with maturity $T$ and strike $K$ namely, a contract with payoff
\begin{equation*}
C(K, T) = \sum_{m=1}^M(S(t_m)-K)^+=\sum_{m=1}^M c_m(K, t_m), \quad t_1, t_2, \dots t_M=T.
\end{equation*}
Such a contract is commonly used for hedging purposes or for the  parameters calibration. It normally encompasses monthly, quarterly and yearly maturities but is not very liquid and  is generally offered by brokers. 

We assume that the market model\refeqq{eq:market} is driven by  a full seven-parameters OU-BCTS  process with infinite activity and finite variation. We price the strip of calls using the FFT-based technique of Carr and Madam\mycite{Carr1999OptionVU} given the \chf\ $\phi(u, t)$ of the of $\log\,S(t)=\log F(0,t) + h(t) + X(t)$
\begin{equation*}
\phi(u, t)= F(0,t)e^{i\,u\,h(t)}\varphi_X(u,t) 
= F(0,t)e^{i\,u\,\left(h(t) + a X(0)\right) + \psi_Z(u, t)}, \quad a = e^{-b\, t}
\end{equation*}
where $h(t)$ is given by\refeqq{eq:rn:ou:bcts} and $\psi_Z(u,t)$ by\refeqq{eq:chf:ou:bcts}. We refer the reader to  Carr and Madam\mycite{Carr1999OptionVU} for the details on the method.	

The calibration and the parameters estimation is not the focus of this study, instead we rather illustrate the applicability of our theoretical results taking parameters sets available in the literature. In this example, we take those of Poirot and Tankov\mycite{PoirotTankov2007} (plus $b$ and $c_n$) and let $\alpha_p$ and $\alpha_n$ vary: $(b, \beta_p, \beta_n, c_p, c_n) = (0.1, 2.5, 3.5, 0.5, 1)$; for simplicity we consider a flat forward curve with $F(0,t)=20, t>0$.  

\begin{table}[ht!]
\centering
\begin{tabular}{|c||*{5}{c|}}
\hline
\backslashbox{$\alpha_n$}{$\alpha_p$}
		&\makebox[3em]{0.1}&\makebox[3em]{0.3}&\makebox[3em]{0.5}
		&\makebox[3em]{0.7}&\makebox[3em]{0.9}\\\hline\hline
			$0.1$ & $3.504$ & $3.540$ & $3.609$ & $4.262$ & $5.770$\\
			$0.3$ & $4.865$ & $4.917$ & $5.008$ & $5.205$ & $6.290$\\
			$0.5$ & $6.690$ & $6.757$ & $6.869$ & $7.073$ & $7.560$\\
			$0.7$ & $9.058$ & $9.136$ & $9.261$ & $9.474$ & $9.879$\\
			$0.9$ & $12.108$ & $12.192$ & $12.322$ & $12.535$ & $12.907$\\
                    \hline
        \end{tabular}
    \scriptsize
    \caption{\footnotesize{Strip of $M=	30$ daily call options calculated with FFT, $T=1/12$  $(b, \beta_p, c_p)=(0.1, 1.5, 0.3)$.}}\label{tab:fft:call:alpha}
\end{table}

Table\myref{tab:fft:call:alpha} shows the values relatively to a strip of $M=30$ daily at-the-money call options with maturity $T=1/12$ with different pairs of $\alpha_p, \alpha_n$. We observe that fixing one of $\alpha_p$ or $\alpha_n$, the value of the option is increasing when the other one increases. 
Moreover, Figure\myref{fig:ou:bcts:call:strike} illustrates the variability of the option price with respect to the strike price $K$ where the dotting lines represent the values obtained with $N=10^5$ MC simulations plus and minus three times the estimation error (the root-mean squared error divided by $\sqrt{N}$). In addition, Figure\myref{fig:ou:bcts:call:maturity} compares the price of at-the-money options $C_m = c_m(K, t_m), m=1, \dots, M$, $K=20$ obtained with the FFT method to those estimated once again with $N=10^5$ MC simulations. In these last two examples we have selected $\alpha_p=\alpha_n=0.5$.

As far as the MC method is concerned, the simulation of the skeleton of the process is accomplished running the procedure explained in Cufaro Petroni and Sabino\mycite{cs20_3} based on Proposition\myref{eq:sub} two times because of the bilateral OU-BCTS; the acceptance rejection step to draw from the law of $V$ in\refeqq{eq:ou:cts:mixture:rate} assumes a piece-wise approximation of the dominating functions into $L=100$ terms.

The results calculated with the FFT-method and with the MC method are totally consistent on the other hand, it is well-know that the FFT approach is faster. Nevertheless, a side-product of the MC approach are percentiles or other statistics which are widely used by practitioners for risk-management purposes. 

\begin{figure}
        \begin{subfigure}[c]{.5\textwidth}{
                \includegraphics[width=70mm]{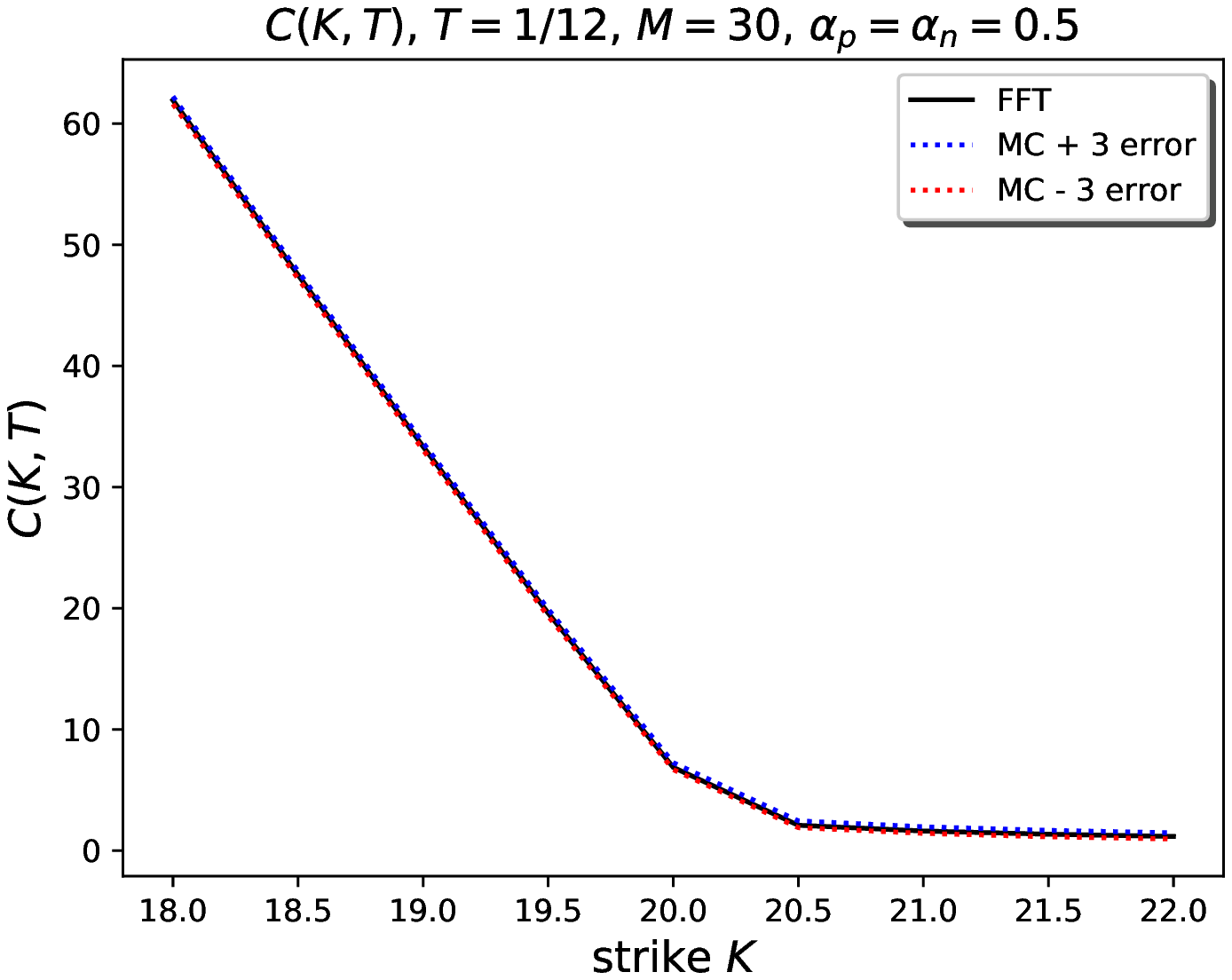}
                }
								\caption{Effect on the strike}\label{fig:ou:bcts:call:strike}
        \end{subfigure}
        \begin{subfigure}[c]{.5\textwidth}{
                \includegraphics[width=70mm]{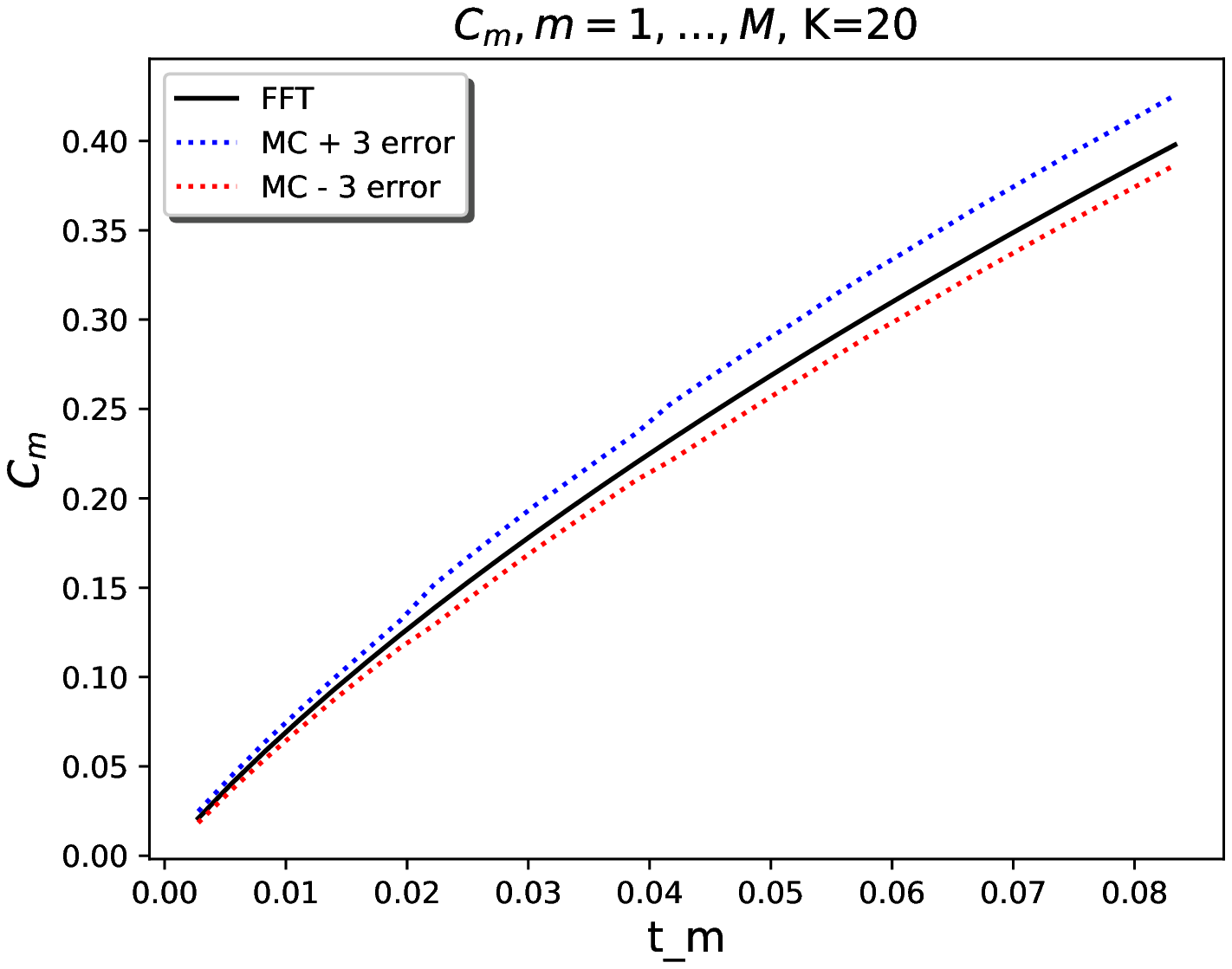}
                }
								\caption{$c_m(K, t_m), m=1, \dots, 30$}\label{fig:ou:bcts:call:maturity}
        \end{subfigure}
\caption{Call-options values calculated with FFT and MC with $N=10^5$, $(b, \beta_p, \beta_n, c_p, c_n) = (0.1, 2.5, 3.5, 0.5, 1)$, $\alpha_p=\alpha_n=0.5$, $K=20$}
\end{figure}%
%
%

\subsection{Asian Options}\label{subsec:fin:app:asian}
As a second financial application we consider the pricing of Asian options with MC simulations. In contrast to the previous example, we assume that the market dynamic is driven by a OU-CGMY process with infinite activity and finite variation with $C=c_p = c_n$, $G=\beta_n$, $M=\beta_p$ and $Y = \alpha_p=\alpha_n$.

MC methods are known to be sometimes slower than FFT techniques that can also be tailored to the pricing of Asian options (see Zhang and C. Oosterlee\mycite{ZhangOosterlee13_a}). Nevertheless, the former approach  provides a view on the distribution of the potential cash-flows of derivative contracts giving a precious information to risk managers or to trading units. 

We recall that the payoff at maturity $T$ of an Asian option with European style and strike price $K$ is
\begin{equation*}
	A(K, T) = \left(\frac{\sum_{i=1}^{I}}{I}S(t_i) - K\right)^+.
\end{equation*}

In this second example we consider once again a flat forward curve $F(0, t)=20, t>0$ and a different parameter set. We select $b=10$ and $\left(C, G, M\right)=\left(2, 15, 5\right)$ also used in Ballotta and Kyriakou\mycite{BK14} and let $Y$ vary. Figure\myref{fig:ou:cgmy:positive:trajectories} displays a sample of four trajectories with these parameters generated using the procedure of Cufaro Petroni and Sabino\mycite{cs20_3} as done in the case of the daily strips of call options.

In addition to this simulation procedure we consider here two approximations: the
first boils down to simply neglect $X_2$ in the
Proposition\myref{prop:ou:ts} and accordingly to the bilateral case (\emph{Approximation 1}); the second -- in the same vein of Benth
et al.\mycite{BDPL18} dealing with the \emph{normal inverse
Gaussian}-driven OU processes -- takes advantage of the
approximation of the law of $Z(t)$ in\refeqq{eq:sol:OU} with that of
$e^{-k\, t}L(t)$ (\emph{Approximation 2}). 

\noindent It turns our that under \emph{Approximation 1} $X_1(\Delta\,t)\sim\bcts\!\left(Y, Y, \frac{M}{a}, \frac{G}{a}, C\,\frac{1 - a^Y}{Y\, b}, C\,\frac{1 - a^Y}{Y\, b}\right)$ whereas under \emph{Approximation 2} $X_1(\Delta\,t)\sim\bcts\left(Y, Y, \frac{M}{a}, \frac{G}{a},
C\,\Delta\,t, C\,\Delta\,t\right)$, $a=e^{-b\,\Delta\,t}$ where $\Delta\,t=1/360$ because we assume daily settlements; for simplicity we adopt the convention that there are $360$ days per year.

In order to highlight the differences between the estimations returned by the exact method and those with the two alternatives, we consider two Asian options both of them with $I=90$ daily settlements, the second option however however, is a forward start contract whose first settlement date occurs after $30$ days. The MC option values and their relative errors are reported in Table\myref{tab:Asian} and Table\myref{tab:ForwardAsian} with different $Y$'s and number of simulations $N$.

Irrespective to the combination of $Y$ and $N$, for the option that start settling after one day the exact solution and \emph{Approximation 1} return very close values, whereas \emph{Approximation 1} is slightly biased. In contrast, for the forward start contract, although the time steps for $m>1$ coincide and are very small, for the simple fact that the first time step is relatively high, the estimated prices returned by the two non-exact simulation schemes  are very biased and do not offer an acceptable alternative any longer. More important, the bias cannot be controlled increasing the number of simulations as shown in Table\myref{tab:ForwardAsian}.

The cause of this difference comes from the fact that $X_2$ in Proposition\myref{prop:ou:ts} can be neglected when the time step is small (accordingly for the bilateral case). Indeed taking the Taylor expansion of the parameter $\Lambda_a$ in Proposition\myref{prop:ou:ts}
\begin{equation*}
    \Lambda_a=\frac{c\Gamma(1-\alpha)b\,\beta^{\alpha}}{2}\,
\Delta\, t^2+o\big(\Delta\,t^2\big).
\end{equation*}
As mentioned, the parameters calibration is not the focus of this study, nevertheless these observations could lead to a convenient strategy combining
parameters estimation and exact simulation of the OU-BCTS processes.
Assuming that the data could be made available with a fine enough
time-granularity (e.g. daily $\Delta\, t=1/360$), one could base the
parameters estimation on the likelihood methods by approximating the
exact transition \pdf\ of a OU-BCTS process with that of a BCTS law. In alternative, one could also use the generalized method of moments to historical data taking the cumulants from the formulas\refeqq{eq:cumulants:ou2} and\refeqq{eq:cts:cumulants}. Instead, to avoid being
forced to always simulate the OU-BCTS processes on a fine time-grid
allowing the approximations, the generation of the
skeleton of such processes will be preferably based on the exact
method.

\begin{figure}
\caption{Sample trajectories of OU-CGMY processes with $\left(b, C,
G, M\right) = \left(10, 2, 15, 5\right)$ and $Y\in\{0.3, 0.5,
0.7, 0.9\}$}\label{fig:ou:cgmy:positive:trajectories}
        \begin{subfigure}[c]{.5\textwidth}{
                \includegraphics[width=70mm]{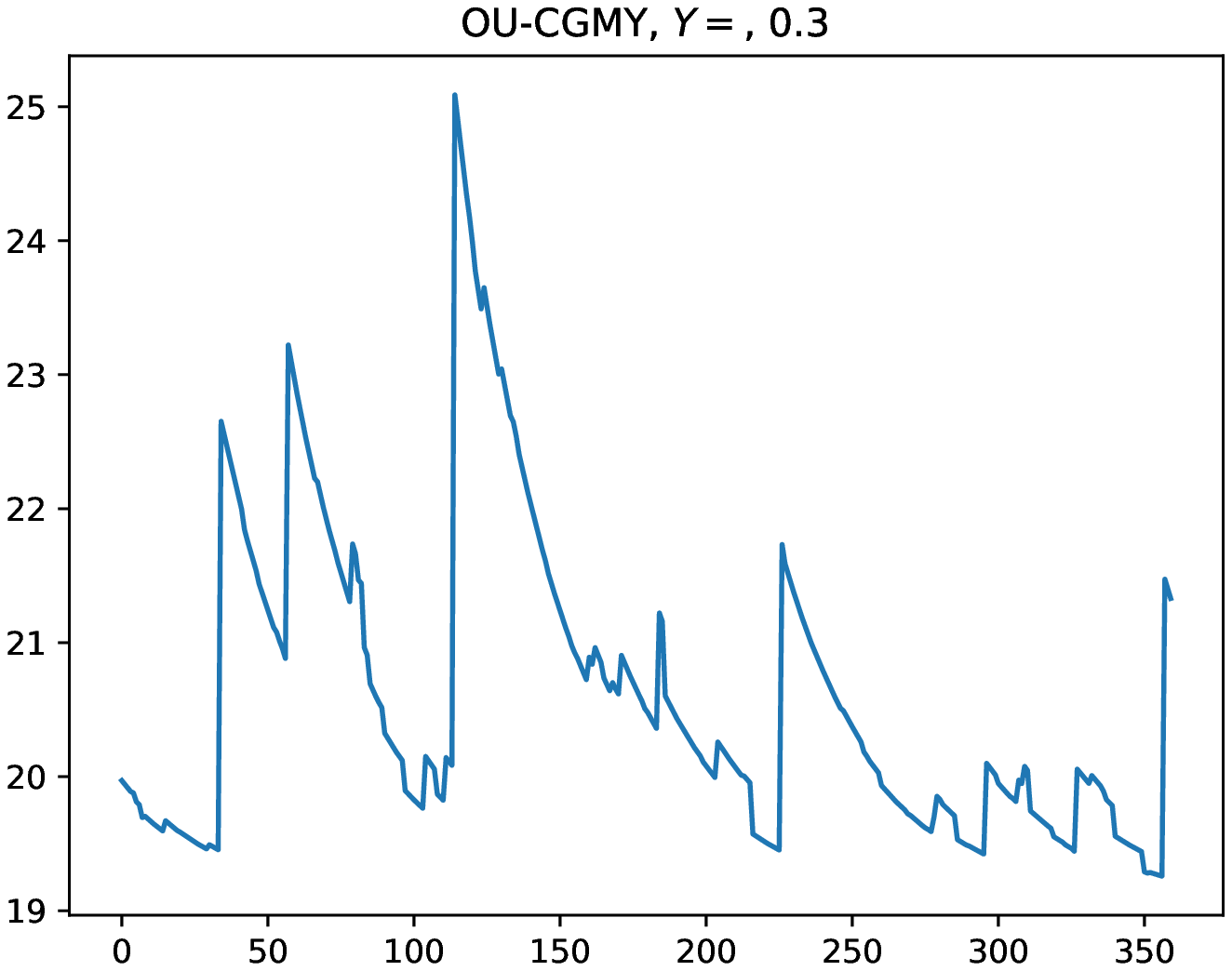}
                }
        \end{subfigure}
        \begin{subfigure}[c]{.5\textwidth}{
                \includegraphics[width=70mm]{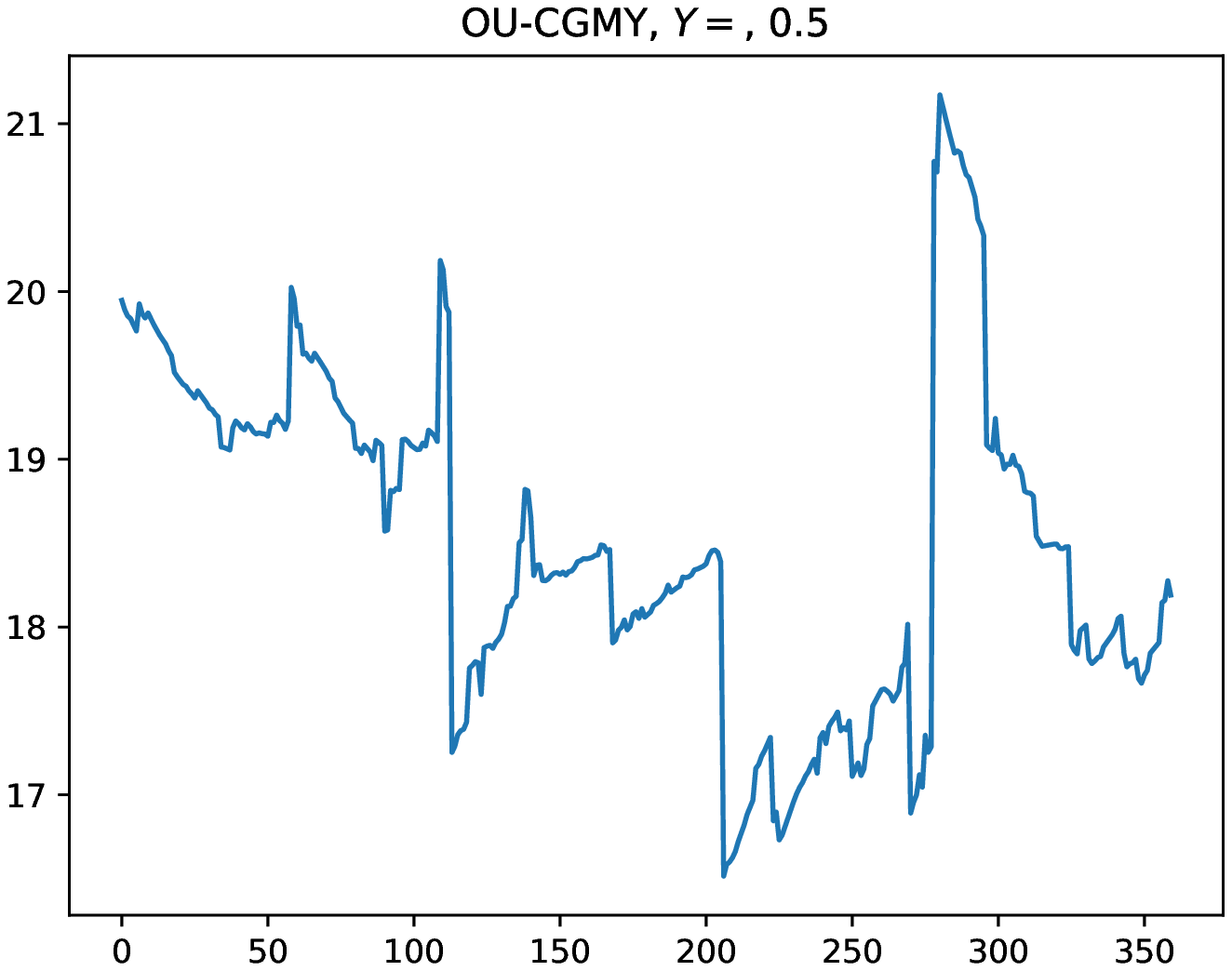}
                }
        \end{subfigure}
			\\
        \begin{subfigure}[c]{.5\textwidth}{
                \includegraphics[width=70mm]{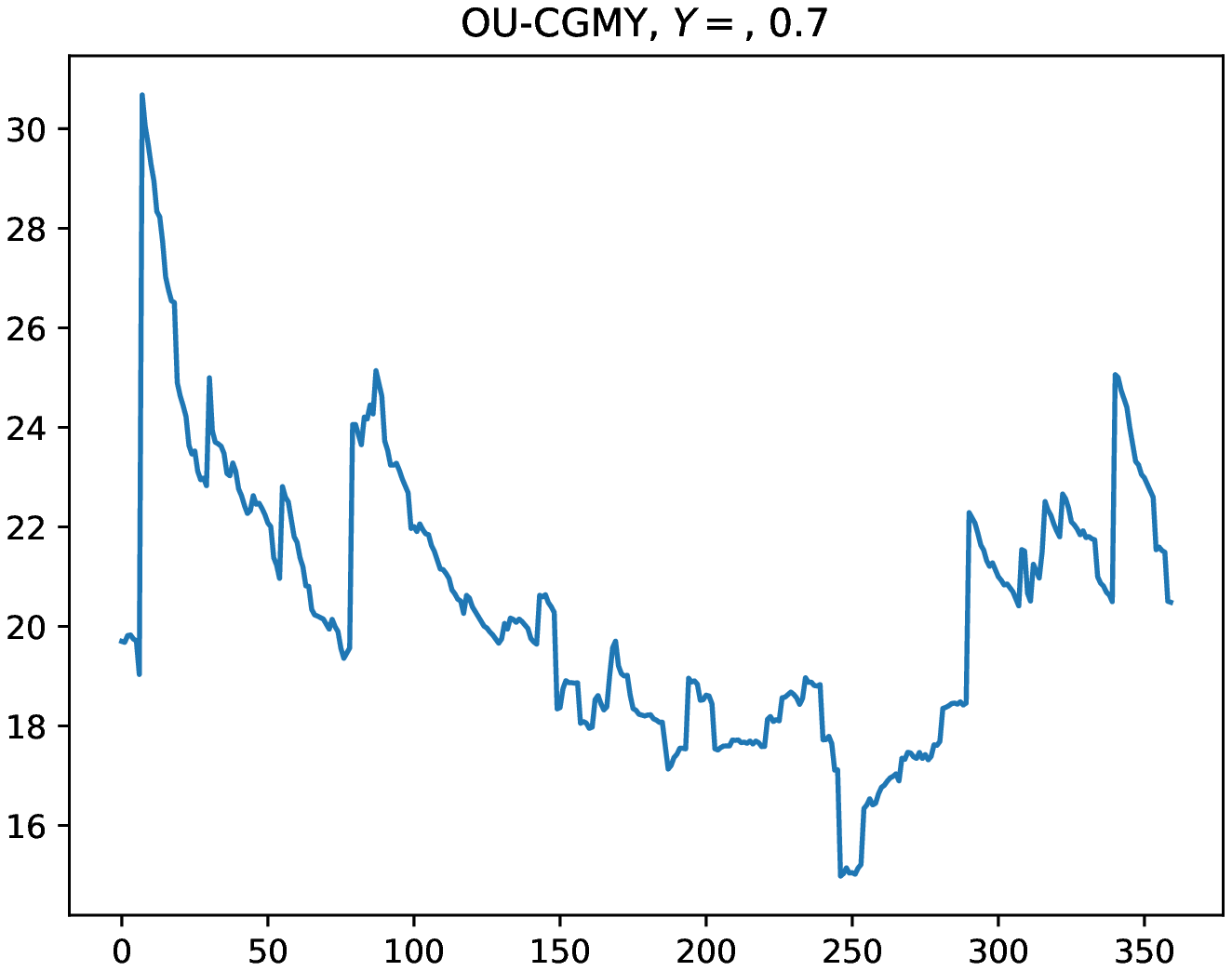}
                }
        \end{subfigure}
        \begin{subfigure}[c]{.5\textwidth}{
                \includegraphics[width=70mm]{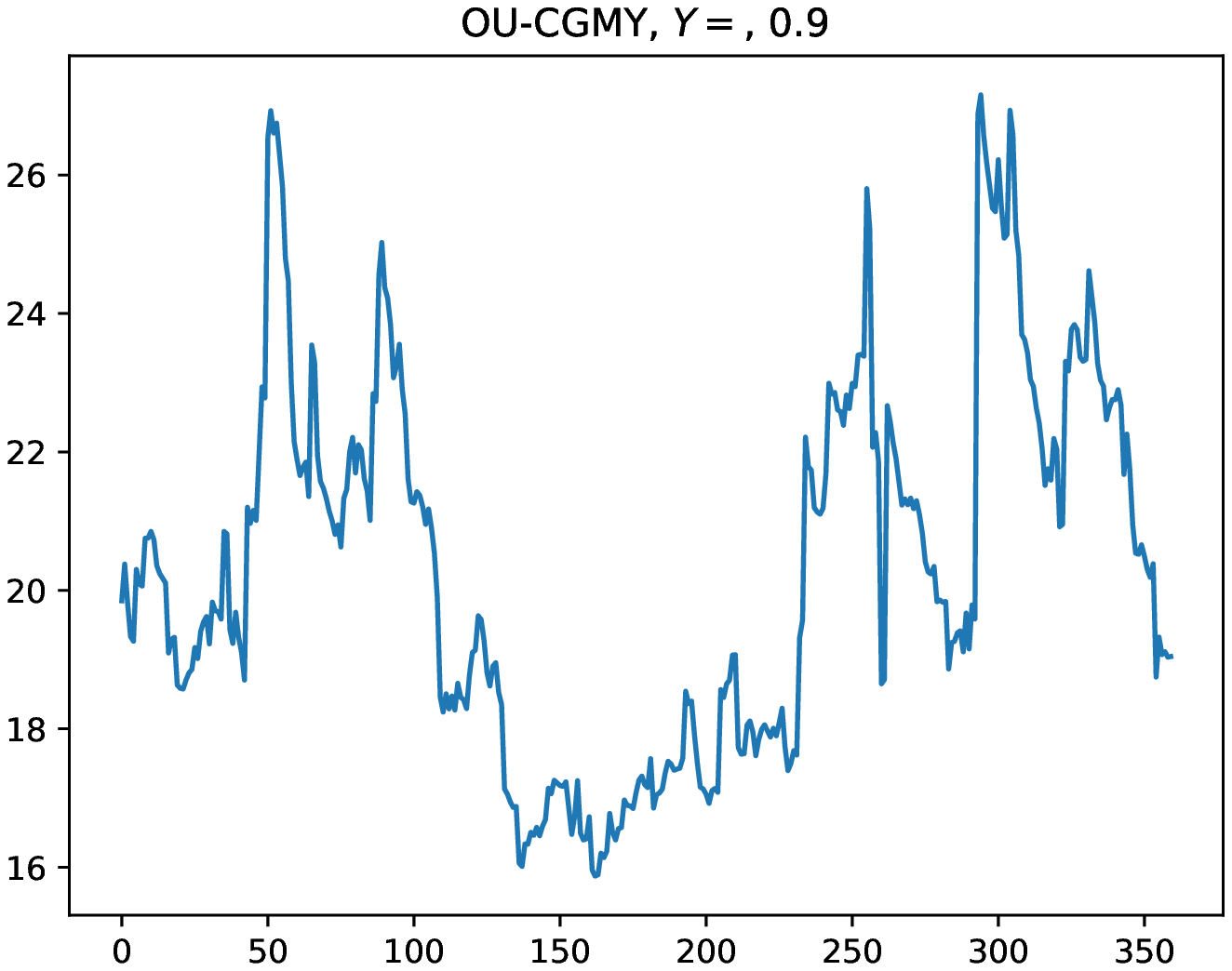}
                }
        \end{subfigure}
\end{figure}

\begin{table}[!htb]
\caption{Asian option. $K=20$, $T=1/4$}\label{tab:Asian}
\begin{subtable}{\textwidth}
		\centering\scriptsize
			\begin{tabular}{|c|c|c|c|c|c|c|c|c|c|c|c|c|}
				\hline
				   & \multicolumn{2}{|c}{Exact} & \multicolumn{2}{|c}{Approximation 1} & \multicolumn{2}{|c}{Approximation 2}
				   & \multicolumn{2}{|c}{Exact} & \multicolumn{2}{|c}{Approximation 1} & \multicolumn{2}{|c|}{Approximation 2}\\
				\hline
				$N$ & price & error & price & error & price & error & price & error & price & error & price & error\\
				\hline
				   & \multicolumn{6}{|c}{$Y=0.3$} & \multicolumn{6}{|c|}{$Y=0.5$}\\
				\hline
$1000$ & $0.3293$ & $0.0224$ & $0.3771$ & $0.0275$ & $0.3671$ & $0.0253$ & $0.4439$ & $0.0278$ & $0.4461$ & $0.0300$ & $0.4865$ & $0.0326$\\
$10000$ & $0.3894$ & $0.0087$ & $0.3709$ & $0.0085$ & $0.3655$ & $0.0085$ & $0.4699$ & $0.0098$ & $0.4673$ & $0.0096$ & $0.4499$ & $0.0096$\\
$20000$ & $0.3692$ & $0.0060$ & $0.3743$ & $0.0062$ & $0.3651$ & $0.0060$ & $0.4571$ & $0.0069$ & $0.4404$ & $0.0065$ & $0.4552$ & $0.0068$\\
$50000$ & $0.3808$ & $0.0040$ & $0.3814$ & $0.0039$ & $0.3718$ & $0.0038$ & $0.4634$ & $0.0043$ & $0.4594$ & $0.0043$ & $0.4545$ & $0.0043$\\
$100000$ & $0.3792$ & $0.0028$ & $0.3715$ & $0.0027$ & $0.3697$ & $0.0027$ & $0.4627$ & $0.0031$ & $0.4559$ & $0.0031$ & $0.4467$ & $0.0030$\\			
				\hline
				\hline
				   & \multicolumn{2}{|c}{Exact} & \multicolumn{2}{|c}{Approximation 1} & \multicolumn{2}{|c}{Approximation 2}
				   & \multicolumn{2}{|c}{Exact} & \multicolumn{2}{|c}{Approximation 1} & \multicolumn{2}{|c|}{Approximation 2}\\
				\hline
				$N$ & price & error & price & error & price & error & price & error & price & error & price & error\\
				\hline
				   & \multicolumn{6}{|c}{$Y=0.7$} & \multicolumn{6}{|c|}{$Y=0.9$}\\
				\hline
$1000$ & $0.5446$ & $0.0325$ & $0.5317$ & $0.0332$ & $0.5298$ & $0.0325$ & $0.6910$ & $0.0427$ & $0.7031$ & $0.0377$ & $0.7246$ & $0.0395$\\
$10000$ & $0.5902$ & $0.0114$ & $0.5630$ & $0.0110$ & $0.5545$ & $0.0106$ & $0.7102$ & $0.0130$ & $0.7227$ & $0.0134$ & $0.6912$ & $0.0129$\\
$20000$ & $0.5851$ & $0.0080$ & $0.5631$ & $0.0075$ & $0.5544$ & $0.0077$ & $0.7147$ & $0.0091$ & $0.7066$ & $0.0091$ & $0.6973$ & $0.0090$\\
$50000$ & $0.5647$ & $0.0049$ & $0.5631$ & $0.0048$ & $0.5413$ & $0.0047$ & $0.7209$ & $0.0058$ & $0.6933$ & $0.0057$ & $0.6920$ & $0.0057$\\
$100000$ & $0.5701$ & $0.0035$ & $0.5603$ & $0.0034$ & $0.5491$ & $0.0034$ & $0.7148$ & $0.0041$ & $0.7117$ & $0.0041$ & $0.6897$ & $0.0040$\\
				\hline				
		\end{tabular}
\end{subtable}
\end{table}
\begin{table}[!htb]
\caption{Forward start Asian option. $K=20$, $T=1/3$}\label{tab:ForwardAsian}
\begin{subtable}{\textwidth}
		\centering\scriptsize
			\begin{tabular}{|c|c|c|c|c|c|c|c|c|c|c|c|c|}
				\hline
				   & \multicolumn{2}{|c}{Exact} & \multicolumn{2}{|c}{Approximation 1} & \multicolumn{2}{|c}{Approximation 2}
				   & \multicolumn{2}{|c}{Exact} & \multicolumn{2}{|c}{Approximation 1} & \multicolumn{2}{|c|}{Approximation 2}\\
				\hline
				$N$ & price & error & price & error & price & error & price & error & price & error & price & error\\
				\hline
				   & \multicolumn{6}{|c}{$Y=0.3$} & \multicolumn{6}{|c|}{$Y=0.5$}\\
				\hline
$1000$ & $0.4361$ & $0.0314$ & $0.3793$ & $0.0274$ & $0.3536$ & $0.0261$ & $0.4734$ & $0.0300$ & $0.4621$ & $0.0301$ & $0.4226$ & $0.0312$\\
$10000$ & $0.4400$ & $0.0098$ & $0.3751$ & $0.0086$ & $0.3689$ & $0.0089$ & $0.5351$ & $0.0107$ & $0.4527$ & $0.0094$ & $0.4161$ & $0.0091$\\
$20000$ & $0.4522$ & $0.0071$ & $0.3884$ & $0.0062$ & $0.3775$ & $0.0063$ & $0.5296$ & $0.0075$ & $0.4602$ & $0.0067$ & $0.4294$ & $0.0066$\\
$50000$ & $0.4458$ & $0.0043$ & $0.3731$ & $0.0038$ & $0.3694$ & $0.0038$ & $0.5303$ & $0.0048$ & $0.4551$ & $0.0043$ & $0.4340$ & $0.0042$\\
$100000$ & $0.4480$ & $0.0031$ & $0.3766$ & $0.0027$ & $0.3685$ & $0.0027$ & $0.5312$ & $0.0034$ & $0.4632$ & $0.0031$ & $0.4274$ & $0.0030$\\

				\hline
				\hline
				   & \multicolumn{2}{|c}{Exact} & \multicolumn{2}{|c}{Approximation 1} & \multicolumn{2}{|c}{Approximation 2}
				   & \multicolumn{2}{|c}{Exact} & \multicolumn{2}{|c}{Approximation 1} & \multicolumn{2}{|c|}{Approximation 2}\\
				\hline
				$N$ & price & error & price & error & price & error & price & error & price & error & price & error\\
				\hline
				   & \multicolumn{6}{|c}{$Y=0.7$} & \multicolumn{6}{|c|}{$Y=0.9$}\\
				\hline
$1000$ & $0.5575$ & $0.0347$ & $0.6160$ & $0.0376$ & $0.5185$ & $0.0329$ & $0.8454$ & $0.0477$ & $0.7726$ & $0.0429$ & $0.6411$ & $0.0399$\\
$10000$ & $0.6397$ & $0.0120$ & $0.5672$ & $0.0110$ & $0.5268$ & $0.0107$ & $0.7941$ & $0.0141$ & $0.7403$ & $0.0132$ & $0.6129$ & $0.0121$\\
$20000$ & $0.6517$ & $0.0086$ & $0.5789$ & $0.0079$ & $0.4967$ & $0.0072$ & $0.8092$ & $0.0101$ & $0.7450$ & $0.0093$ & $0.6013$ & $0.0086$\\
$50000$ & $0.6574$ & $0.0055$ & $0.5721$ & $0.0050$ & $0.5129$ & $0.0048$ & $0.8124$ & $0.0064$ & $0.7429$ & $0.0060$ & $0.5976$ & $0.0054$\\
$100000$ & $0.6500$ & $0.0039$ & $0.5719$ & $0.0035$ & $0.5094$ & $0.0033$ & $0.8142$ & $0.0046$ & $0.7399$ & $0.0042$ & $0.5970$ & $0.0038$\\

				\hline				
		\end{tabular}
\end{subtable}
\end{table}

\subsection{Swing Options}\label{subsec:fin:app:swing}

A swing option is a type of contract used by investors in energy
markets that lets the option holder buy a predetermined quantity of
energy at a predetermined price (strike), while retaining a certain
degree of flexibility in  both the amount purchased and the
price paid. 

Let the maturity date $T$ be fixed and the payoff at time $t<T$ be given by $(S(t)-K)^+$
where  $K$ denotes the strike price, in addition we assume only one unit of the underlying can be exercised any
time period. Let $V(n, s, t)$ denote the price of such a swing option at time $t$ given the spot price $s$
which has $n$ out of $N$ exercise rights left. For $m=1,\dots, M_S$, the dynamic programming principle allows us to
write (see Bertsekas\mycite{Bertsekas05})
\begin{equation}
V(n, s, t_m) =\max
              \left\{
                   \begin{array}{ll}
									 \EXP{V(n, S(t_{m+1}), t_{m+1})|S(t_m) = s}, \\
                                 \\
									\EXP{V(n, S(t_{m+1}), t_{m+1})|S(t_m) = s} + (s-K)^+
                   \end{array}
             \right\},\quad n < N 
\label{eq:swing}
\end{equation}
and $V(n, s, T) = (S(T) - K)^+$, $n \le N$ and $V(0, s, t)=0$.
In order to solve the recursion equation we rely on the modified
version of the LSMC, introduced in
Longstaff-Schwartz\mycite{LSW01}, detailed in Boogert and C. de Jong\mycite{BDJ08, BDJ10} where the continuation value is approximated with a linear regression with  $m=1,\dots, M_S$
\begin{equation*}
	\EXP{V(n,  S(t_{m+1}), t_{m+1})|S(t_m) = s}\simeq a_0 + a_1 S(t_m) + \dots, +a_B S^B(t_m), \quad n < N.
\end{equation*}
In our experiments, we used simple power polynomials with $B=3$, but the regression may be performed on a different set of basis functions as well (see Boogert and de
Jong\mycite{BDJ10} for a comparison with other basis functions).

Several other approaches have been proposed: for instance one may solve the recursion by adapting the method of Ben-Ameur et al.\mycite{BBKL2007} or might use the quantization technique of Bardou et al.\mycite{BBP07}. In alternative, one can also use the tree method of Jaillet et al.\mycite{JRT04} or the Fourier cosine expansion in Zhang and C. Oosterlee\mycite{ZhangOosterlee13_b} taking advantage of the explicit form of the \chf\ of OU-BCTS process.

In this last example we assume a OU-CGMY driven market model with $Y<0$, namely a combination of mean-reverting compound Poisson processes with positive and negative jumps. We consider a different set of parameters compared to the cases illustrated so far, namely we take $(b, C, G, M)=(25, 80, 10.5, 15.5)$ and let $Y$ vary. The parameters are very different than the other two examples and are chosen to mimic realistic price path as shown in Figure\myref{fig:ou:cgmy:negative:trajectories}. We also remark that, due to the fact that energy markets are very seasonal and spikes occur in clusters due to for instance, cold spells, one could assume that the intensity of the compound Poisson processes is a seasonal time-dependent function. The results in Subsection\myref{subsec:fin:act} and the simulation algorithms in Section\myref{sec:ou:ts:cgmy:alg} can be easily adapted taking a step-wise approximation of the intensity function.

Table\myref{tab:swing} shows the values and MC errors relatively to the pricing of a $120-120$ swing option with maturity $T=1$ and strike price $K=20$, namely the holder has $N=120$ rights and must exercise all of them. We observe that the LSMC combined with Algorithm\myref{alg:ou:cts} produces unbiased results for all selected $Y$'s and apparently $2\times 10^4$ simulations are required to attain an acceptable convergence. In contrast to the Asian option case, it does not make sense to adopt the approximation of the law of $Z(t)$ in\refeqq{eq:sol:OU} with that of
$e^{-k\, t}L(t)$ (\emph{Approximation 2} in Subsection\myref{subsec:fin:app:asian}) because this approach returns another compound Poisson process and therefore does not provide a computational advantage. Overall, it is evident that our
newly developed approach can achieve high accuracy as well as efficiency.

\begin{figure}
\caption{Sample trajectories of OU-CGMY processes with $\left(b, C,
G, M\right) = \left(10, 10, 1.75, 1.25\right)$ and $Y\in\{-0.3, -0.5,
-0.7, -0.9\}$}\label{fig:ou:cgmy:negative:trajectories}
        \begin{subfigure}[c]{.5\textwidth}{
                \includegraphics[width=70mm]{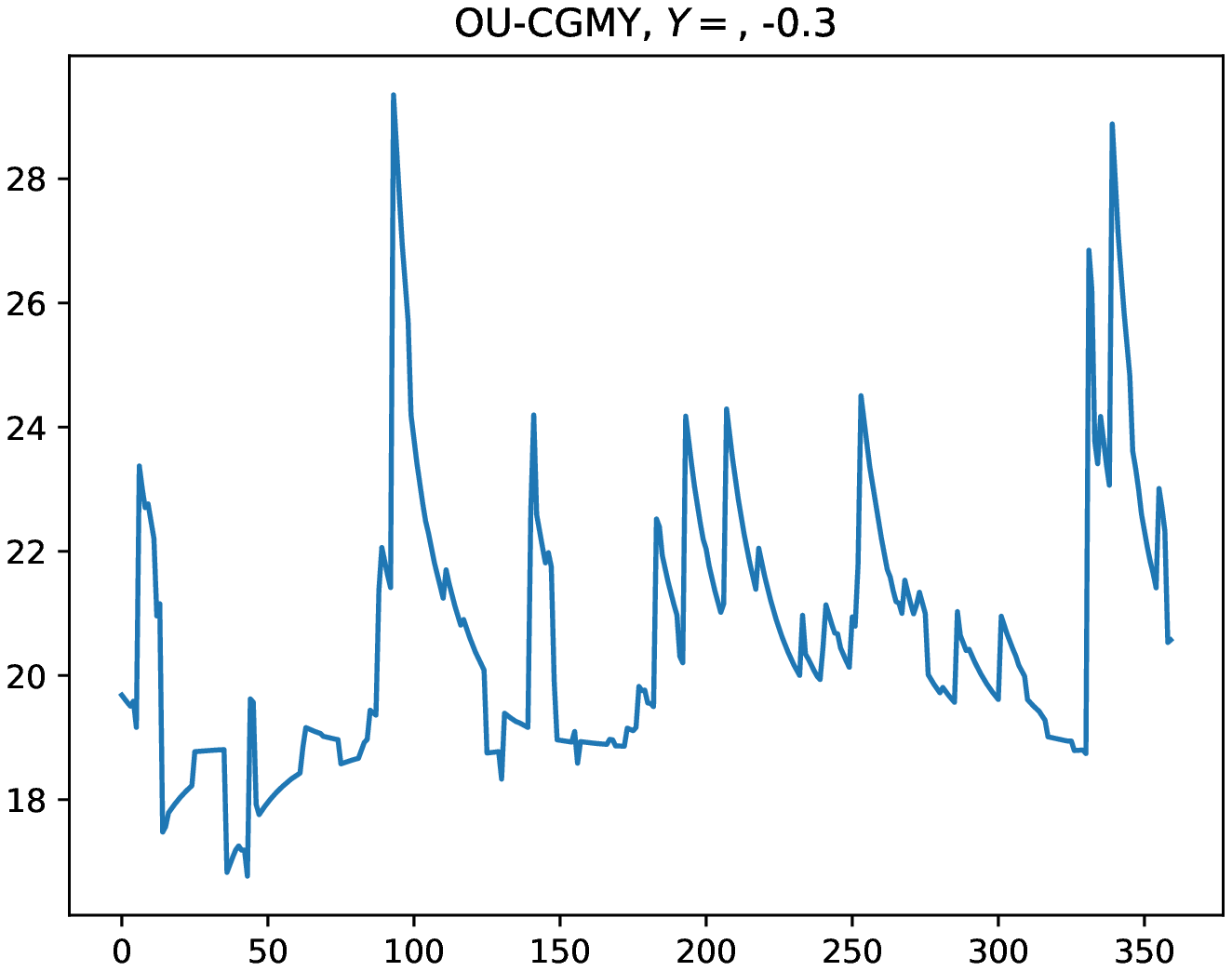}
                }
        \end{subfigure}
        \begin{subfigure}[c]{.5\textwidth}{
                \includegraphics[width=70mm]{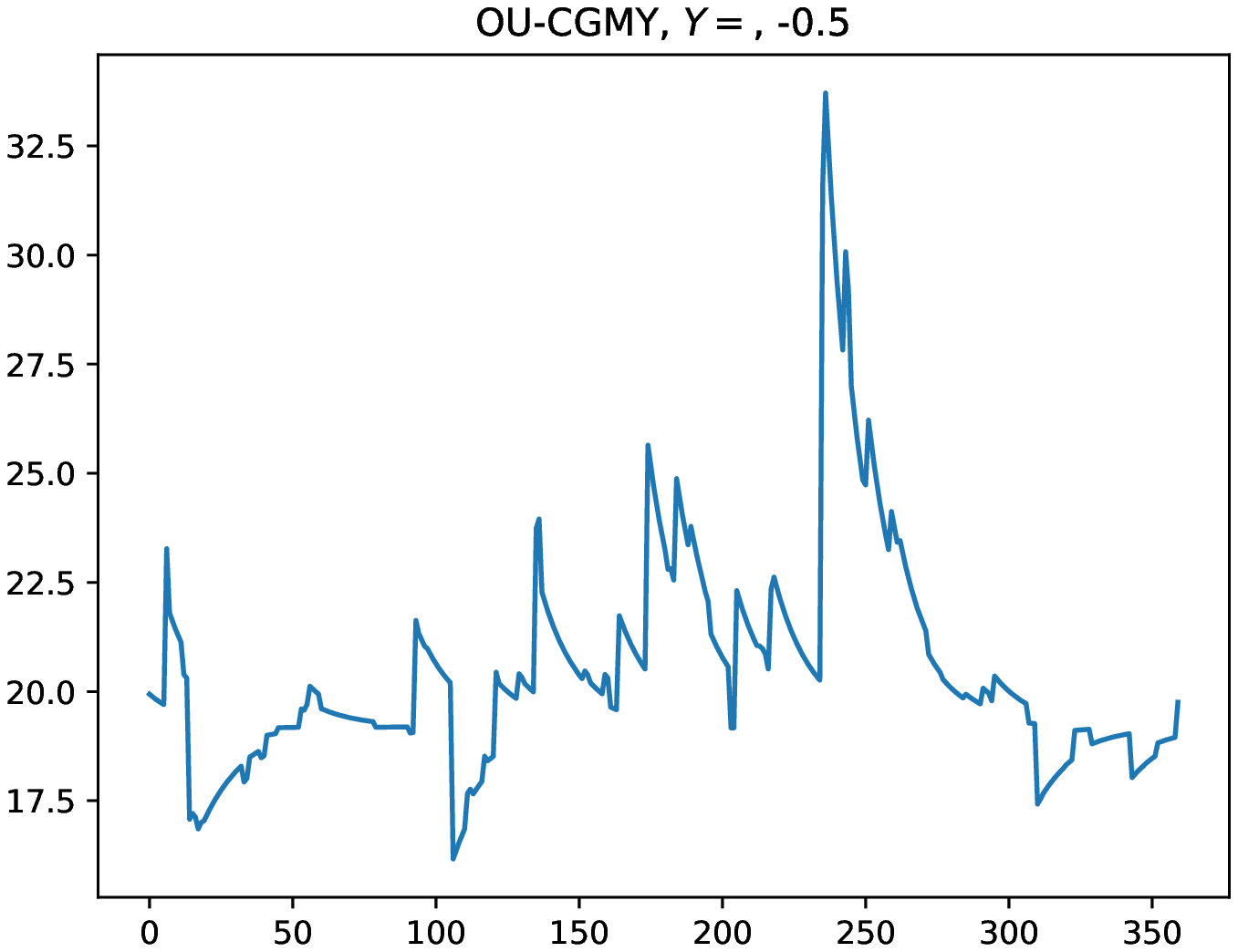}
                }
        \end{subfigure}
			\\
        \begin{subfigure}[c]{.5\textwidth}{
                \includegraphics[width=70mm]{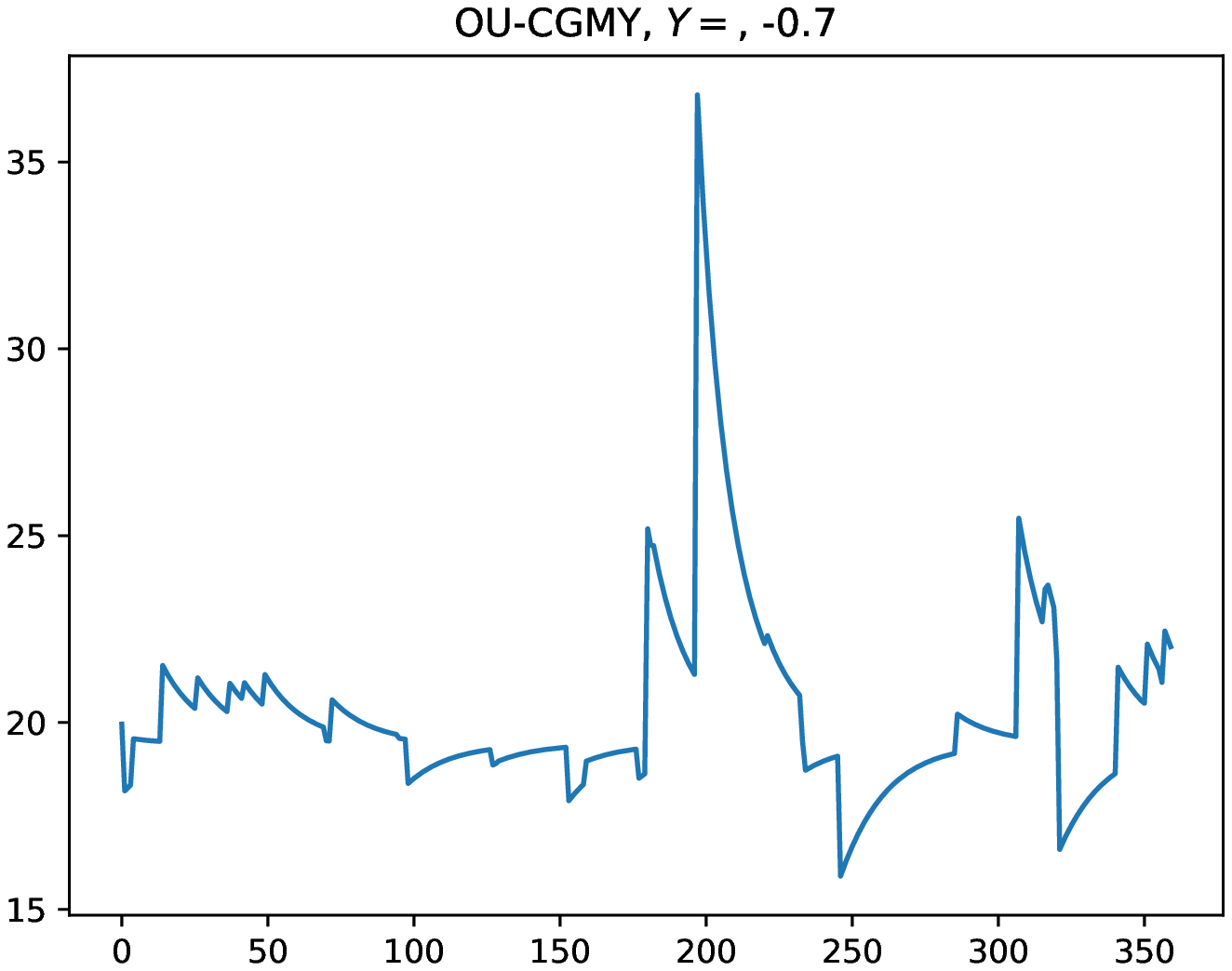}
                }
        \end{subfigure}
        \begin{subfigure}[c]{.5\textwidth}{
                \includegraphics[width=70mm]{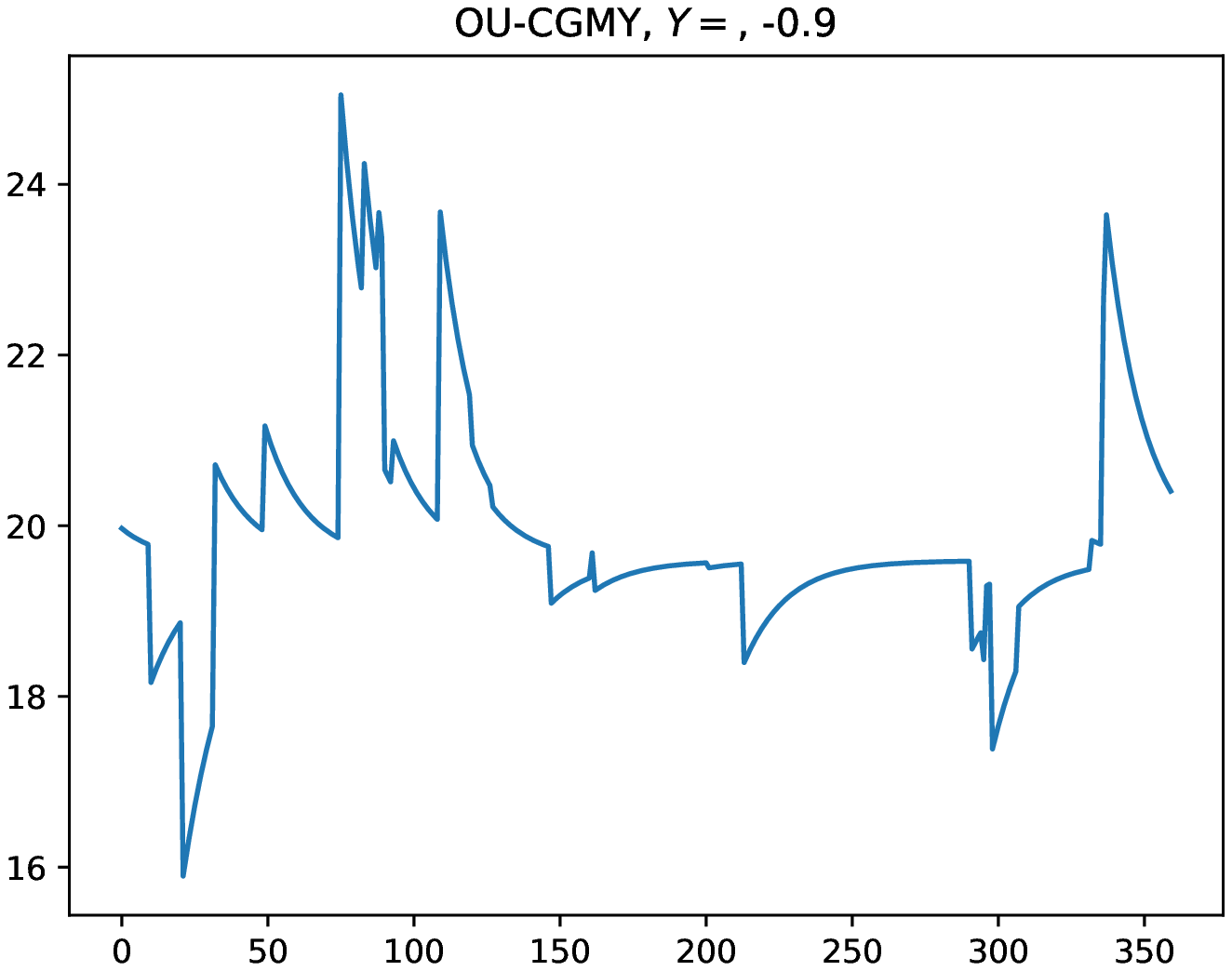}
                }
        \end{subfigure}
\end{figure}%
\begin{table}[!htb]
\caption{120-120 Swing option. $K=20$, $T=1$}\label{tab:swing}
\begin{subtable}{\textwidth}
		\centering\scriptsize
			\begin{tabular}{|c|c|c|c|c|c|c|c|c|}
				\hline
					& \multicolumn{2}{|c}{$Y=0.3$} & \multicolumn{2}{|c|}{$Y=0.5$}
					& \multicolumn{2}{|c}{$Y=0.7$} & \multicolumn{2}{|c|}{$Y=0.9$}\\
				\hline
				$N$ & price & error & price & error & price & error & price & error\\
				\hline
$1000$ & $101.965$ & $2.873$ & $80.500$ & $2.446$ & $66.634$ & $2.222$ & $51.117$ & $1.983$\\
$10000$ & $99.565$ & $0.902$ & $79.848$ & $0.778$ & $63.451$ & $0.682$ & $49.801$ & $0.599$\\
$20000$ & $98.328$ & $0.633$ & $79.774$ & $0.550$ & $64.322$ & $0.489$ & $49.958$ & $0.426$\\
$50000$ & $98.187$ & $0.399$ & $79.585$ & $0.348$ & $63.686$ & $0.307$ & $50.223$ & $0.269$\\
$100000$ & $98.270$ & $0.283$ & $79.284$ & $0.245$ & $63.949$ & $0.217$ & $50.487$ & $0.192$\\
				\hline
		\end{tabular}
\end{subtable}
\end{table}

\subsection{Application to Forward Markets}\label{sec:extensions}
So far, we have discussed the pricing of derivative contracts depending on the day-ahead price which is modeled as a OU process. On the other hand, the main point of Section\myref{sec:ou:ts:cgmy} is the study of the additive process $Z(t)=\int_0^te^{-b\,(t - u}\,dL(u)$ irrespective of the construction of a OU-BCTS or a OU-CGMY process. For instance, it is common practice to model the forward price as a geometric Brownian motion with a time-dependent volatility function that captures the Samuelson effect. For instance, Kiesel et al.\mycite{KSB09} have considered a two-factor market dynamics where one of the two factors depends on $\int_0^te^{-b\,(t - u)}\,dW(u)$ where $W(\cdot)$ is a standard Wiener process.

Beyond the Gaussian world, Piccirilli et al.\mycite{PSV20} have recently proposed a class of models, named Non-Overlapping-Arbitrage models (NOA), with the aim or capturing the Samuelson effect and reproducing the different levels and shapes of the implied volatility profiles displayed by options. 

In particular, they assume a stochastic evolution of a generic future price at time $t$, maturity $T$, $t\le T\le T_1 < T_2$, and with delivery period $[T_1, T_2]$ is described by
\begin{eqnarray}
	F(t, T_1, T_2) &=& F(0, T_1, T_2) + \int_0^t\Gamma_1(u, T_1, T_2) dL_1(u) + \Gamma_2(T_1, T_2)\,X_2(t) \nonumber \\
	&=& F(0, T_1, T_2) + X_1(t, T_1, T_2) + \Gamma_2(T_1, T_2)\,X_2(t)
\label{eq:noa}
\end{eqnarray}
where $L_1(\cdot)$ and $X_2(\cdot)$ are two independent \Levy\ processes. Moreover,
\begin{equation}
\Gamma_1(u, T_1, T_2) = \frac{\gamma_1}{b\,(T_2 - T_1)}\left(e^{-b\,(T_1 -u)} - e^{-b\,(T_2 -u)}\right).
\label{eq:gamma1}
\end{equation}
\begin{equation}
\Gamma(T_1, T_2) = \frac{1}{T_2 - T_1}\int_{T_1}^{T_2}\gamma(u)\,du
\label{eq:gamma2}
\end{equation}
are two deterministic functions that are meant to capture the Samuelson effect in option pricing (see also Jaeck and Lautier \mycite{JL16}). Indeed, in the spirit of Benth et al.\mycite{BPV19} and Latini et al.\mycite{LPV19}, the special form of the coefficients
arises from the implicitly underlying assumption that the future can be written as the average
over an underlying artificial futures price with instantaneous delivery.

Although Piccirilli et al.\mycite{PSV20} illustrate the application of their model under the assumption that $L_1(\cdot)$ and $X_2(\cdot)$ are centered NIG processes, the setting can be modified taking two independent BCTS or CGYM processes. Of course, such models are related to the additive process studied in Section\myref{sec:ou:ts:cgmy}, because, after some algebra it results
\begin{eqnarray*}
	 X_1(t, T_1, T_2) &=& \frac{\gamma_1}{b\,(T_2 - T_1)}\left(e^{-b\,(T_1 -T)} - e^{-b\,(T_2 -
	T)}\right) \int_0^t e^{-b\,(t -u)} dL_1(u)  \\
	&=&  \Gamma_1(T, T_1, T_2) Z(t), 
\end{eqnarray*}
hence the \chf\ and in the particular, the simulation procedure of the skeleton of the additive process $X_1(\cdot, T_1, T_2)$ can be derived from those of $Z(\cdot)$. 
It is worthwhile noticing that the Piccirilli et al.\mycite{PSV20} found an explicit form of the \chf\ of $X_1(t, T_1, T_2)$ when $L_1(\cdot)$ is a centered NIG process, whereas they do not provide any procedure to simulate such a process. In alternative, taking BCTS or CGMY processes and slightly modifying our results, one can get both the \chf\ and the simulation procedure giving the possibility to price other derivative contracts via  Monte Carlo simulations. We omit an explicit proof to avoid overloading the paper with lengthy details.   

\section{Concluding Remarks}\label{sec:conclusions}

In this study we have investigated the pricing of energy derivatives in markets driven by classical tempered stable and CGMY processes of OU type with finite variation. To this end, we have derived the \chf\ of the transition law of such processes in closed form such that we can obtain the non-arbitrage conditions and spot prices that are consistent with the forward curve. In addition, extending the work of Cufaro Petroni and Sabino\mycite{cs20_3}, we have detailed efficient algorithms for the simulation of the skeleton of classical tempered stable and CGMY processes of OU type with particular focus to the case when whey coincide with compound Poisson processes. We have illustrated the applicability of these results to the pricing of three common derivative contracts in energy markets, namely a strip of daily call options, an Asian option with European style and a swing option. 
In our numerical experiments we have selected a one-factor model in order to better highlight the features of our finding, nevertheless the extension to two-factor models in the same vein of Schwartz and Smith\mycite{SchwSchm00} is straightforward.
In the first example we have made use of the explicit knowledge of the \chf\ to implement the pricing with the FFT-based technique of Carr and Madan\mycite{Carr1999OptionVU} and have compared the outcomes to those obtained via MC simulations. In the second example, we have priced Asian options with MC simulations where we have also adopted two common approximations techniques. These approximations provide reliable values if the time steps of the time grid are relatively small but if one considers a forward start contract the outcome is really biased.  Although the parameter calibration and the model selection is not the main focus of this study, these observations give an indication of how one could conceive a simplified procedure for the parameters estimation. In addition, we have shown that the proposed simulation algorithm, combined with the LSMC approach of  Boogert and C. de Jong\mycite{BDJ08, BDJ10}, provides an efficient and accurate pricing of a one year $120-120$ swing option. Furthermore, our results are not restricted to OU processes and to the modeling of spot prices. Indeed, in the spirit of Benth et al.\mycite{BPV19}, Latini et al.\mycite{LPV19} and Piccirilli et al.\mycite{PSV20} they can be adapted to capture the Samuelson effect and to volatility smiles.

Finally, future studies could cover the extension to  a multidimensional
framework for instance adopting the view of Luciano and Semeraro\mycite{SL2010}, Ballotta and Bonfiglioli\mycite{BB2013} or the recent approaches of Gardini et al.\mycite{Gardini20b, Gardini20a} and Lu\mycite{lu2020}. A last topic deserving
further investigation is the time-reversal simulation of the OU processes generalizing the results of Pellegrino and
Sabino\mycite{PellegrinoSabino15} and Sabino\mycite{Sabino20a} to the
case of  classical tempered stable and CGMY processes.


\section*{Acknowledgements}
I would like to express my gratitude to Matteo Gardini and Nicola Cufaro Petroni for their help relatively to the application of the FFT method.

        \bibliographystyle{plain}
        \bibliography{biblioAll}

\begin{thebibliography}{10}

\bibitem{BB2013}
L.~Ballotta and E.~Bonfiglioli.
\newblock {Multivariate Asset Models Using {L}\'{e}vy Processes and
  Applications}.
\newblock {\em The European Journal of Finance}, 13(22):1320--1350, 2013.

\bibitem{BK14}
L.~Ballotta and I.~Kyriakou.
\newblock Monte {C}arlo {S}imulation of the {CGMY} {P}rocess and {O}ption
  {P}ricing.
\newblock {\em Journal of Futures Markets}, 34(12):1095--1121, 2014.

\bibitem{BBP07}
O.~Bardou, S.~Bouthemy, and G.~Pag\'es.
\newblock Optimal {Q}uantization for the {P}ricing of {S}wing {O}ptions.
\newblock {\em Applied Mathematical Finance}, 16(2):183--217, 2009.

\bibitem{BJS1998}
O.~E. Barndorff-Nielsen, J.~L. Jensen, and M.~S{\o}rensen.
\newblock Some {S}tationary {P}rocesses in {D}iscrete and {C}ontinuous {T}ime.
\newblock {\em Advances in Applied Probability}, 30(4):989–1007, 1998.

\bibitem{BNSh01}
O.E. Barndorff-Nielsen and N.~Shephard.
\newblock Non-{G}aussian {O}rnstein-{U}hlenbeck-based {M}odels and some of
  their {U}ses in {F}inancial {E}conomics.
\newblock {\em Journal of the Royal Statistical Society: Series B},
  63(2):167--241, 2001.

\bibitem{BBKL2007}
H.~Ben-Ameur, M.~Breton, L.~Karoui, and P.~L'Ecuyer.
\newblock {A {D}ynamic {P}rogramming {A}pproach for {P}ricing {O}ptions
  {E}mbedded in {B}onds}.
\newblock {\em Journal of Economic Dynamics and Control}, 31(7):2212--2233,
  July 2007.

\bibitem{BKM07}
F.E. Benth, J.~Kallsen, and T.~Meyer-Brandis.
\newblock A non-gaussian ornstein-uhlenbeck process for electricity spot price
  modeling and derivatives pricing.
\newblock {\em Applied Mathematical Finance}, 14(2):153--169, 2007.

\bibitem{BDPL18}
F.E. Benth, L.~Di Persio, and S.~Lavagnini.
\newblock Stochastic {M}odeling of {W}ind {D}erivatives in {E}nergy {M}arkets.
\newblock {\em Risks, MDPI, Open Access Journal}, 6(2):1--21, 2018.

\bibitem{BPV19}
F.E. Benth, M.~Piccirilli, and T.~Vargiolu.
\newblock Mean-reverting {A}dditive {E}nergy {F}orward {C}urves in a
  {H}eath–{J}arrow–{M}orton {F}ramework.
\newblock {\em Mathematics and Financial Economics}, 13:543--577, 2019.

\bibitem{BenthBenth04}
F.E. Benth and J.~\v{S}altyt\'e Benth.
\newblock The {N}ormal {I}nverse {G}aussian {D}istribution and {S}pot {P}rice
  {M}odelling in {E}nergy {M}arkets.
\newblock {\em International Journal of Theoretical and Applied Finance},
  07(02):177--192, 2004.

\bibitem{Bertsekas05}
D.~P. Bertsekas.
\newblock {\em Dynamic {P}rogramming and {O}ptimal {C}ontrol, {V}olume {I}}.
\newblock Athena Scientific, Belmont, Mass., third edition, 2005.

\bibitem{BDJ08}
A.~Boogert and C.~de~Jong.
\newblock Gas {S}torage {V}aluation {U}sing a {M}onte {C}arlo {M}ethod.
\newblock {\em Journal of Derivatives}, 15:81--91, 2008.

\bibitem{BDJ10}
A.~Boogert and C.~de~Jong.
\newblock Gas {S}torage {V}aluation {u}sing a {M}ultifactor {P}rice {M}odel.
\newblock {\em The Journal of Energy Markets}, 4:29--52, 2011.

\bibitem{CarrCrosby10}
P.~Carr and J.~Crosby.
\newblock A {C}lass of {L}\'evy {P}rocess {M}odels with almost {E}xact
  {C}alibration to both {B}arrier and {V}anilla {FX} {O}ptions.
\newblock {\em Quantitative Finance}, 10(10):1115--1136, 2010.

\bibitem{CGMY2002}
P.~Carr, H.~Geman, D.B. Madan, and M.~Yor.
\newblock The fine structure of asset returns: An empirical investigation.
\newblock {\em The Journal of Business}, 75(2):305--332, 2002.

\bibitem{Carr1999OptionVU}
P.~Carr and D.B. Madan.
\newblock Option {V}aluation {U}sing the {F}ast {F}ourier {T}ransform.
\newblock {\em Journal of Computational Finance}, 2:61--73, 1999.

\bibitem{CarteaFigueroa}
A.~Cartea and M.~Figueroa.
\newblock Pricing in {E}lectricity {M}arkets: a {M}ean {R}everting {J}ump
  {D}iffusion {M}odel with {S}easonality.
\newblock {\em Applied Mathematical Finance, No. 4, December 2005},
  12(4):313--335, 2005.

\bibitem{ContTankov2004}
R.~Cont and P.~Tankov.
\newblock {\em Financial {M}odelling with {J}ump {P}rocesses}.
\newblock Chapman and Hall, London, 2004.

\bibitem{Cufaro08}
N.~{Cufaro Petroni}.
\newblock Self-decomposability and {S}elf-similarity: a {C}oncise {P}rimer.
\newblock {\em Physica A, Statistical Mechanics and its Applications},
  387(7-9):1875--1894, 2008.

\bibitem{cs20_2}
N.~{Cufaro Petroni} and P.~Sabino.
\newblock Fast {P}ricing of {E}nergy {D}erivatives with {M}ean-reverting
  {J}ump-diffusion {P}rocesses.
\newblock Available at: https://arxiv.org/abs/1908.03137.

\bibitem{cs20_3}
N.~{Cufaro Petroni} and P.~Sabino.
\newblock Tempered {S}table {D}istribution and {F}inite {V}ariation
  {O}rnstein-{U}hlenbeck {P}rocesses.
\newblock Available at: https://arxiv.org/abs/2011.09147.

\bibitem{Gardini20b}
M.~Gardini, P.~Sabino, and E.~Sasso.
\newblock A {B}ivariate {N}ormal {I}nverse {G}aussian {P}rocess with
  {S}tochastic {D}elay: {E}fficient {S}imulations and {A}pplications to
  {E}nergy {M}arkets, 2020.
\newblock Available at www.arxiv.org.

\bibitem{Gardini20a}
M.~Gardini, P.~Sabino, and E.~Sasso.
\newblock Correlating {L}\'evy {P}rocesses with {S}elf-decomposability:
  {A}pplications to {E}nergy {M}arkets, 2020.
\newblock Available at www.arxiv.org.

\bibitem{Grabchak16}
M.~Grabchak.
\newblock {\em Tempered {S}table {D}istributions}.
\newblock Springer International Publishing, 2016.

\bibitem{gradshteyn2007}
I.~S. Gradshteyn and I.~M. Ryzhik.
\newblock {\em Table of {I}ntegrals, {S}eries, and {P}roducts}.
\newblock Elsevier/Academic Press, Amsterdam, seventh edition, 2007.

\bibitem{HHM11}
B.~Hambly, S.~Howison, and T.~Kluge.
\newblock Information-{B}ased {M}odels for {F}inance and {I}nsurance.
\newblock {\em Quantitative Finance}, 9(8):937--949, 2009.

\bibitem{JL16}
E.~Jaeck and D.~Lautier.
\newblock Volatility in {E}lectricity {D}erivative {M}arkets: The {S}amuelson
  {E}ffect {R}evisited.
\newblock {\em Energy Economics}, 59:300--313, 2016.

\bibitem{JRT04}
P.~Jaillet, E.I. Ronn, and S.~Tompaidis.
\newblock Valuation of {C}ommodity-{B}ased {S}wing {O}ptions.
\newblock {\em Management Science}, 50(7):909--921, 2004.

\bibitem{KT06}
J.Kallsen and P.~Tankov.
\newblock Characterization of {D}ependence of {M}ultidimensional {L}\'{e}vy
  {P}rocesses {U}sing {L}\'{e}vy {C}opulas.
\newblock {\em Journal of Multivariate Analysis}, 97(7):1551--1572, 2006.

\bibitem{Jorgensen97}
B.~J{\o}rgensen.
\newblock {\em The {T}heory of {D}ispersion {M}odels}.
\newblock Chapman \& Hall, 1997.

\bibitem{KSB09}
R.~Kiesel, G.~Schindlmayr, and R.H. Börger.
\newblock A {T}wo-factor {M}odel for the {E}lectricity {F}orward {M}arket.
\newblock {\em Quantitative Finance}, 9(3):279--287, 2009.

\bibitem{Koponen95}
I.~Koponen.
\newblock Analytic {A}pproach to the {P}roblem of {C}onvergence of {T}runcated
  {L}\'evy {F}lights {T}owards the {G}aussian {S}tochastic {P}rocess.
\newblock {\em Phys. Rev. E}, 52:1197--1199, Jul 1995.

\bibitem{KT2013}
U.~K\"{u}chler and S.~Tappe.
\newblock Tempered {S}table {D}istribution and {P}rocesses.
\newblock {\em Stochastic Processes and their Applications}, 123(12):4256 --
  4293, 2013.

\bibitem{LPV19}
L.~Latini, M.~Piccirilli, and T.~Vargiolu.
\newblock Mean-reverting {N}o-arbitrage {A}dditive {M}odels for {F}orward
  {C}urves in {E}nergy {M}arkets.
\newblock {\em Energy Economics}, 79:157--170, 2019.
\newblock Energy Markets Dynamics in a Changing Environment.

\bibitem{L80}
A.J Lawrance.
\newblock Some {A}utoregressive {M}odels for {P}oint {P}rocesses.
\newblock In P.~Bartfai and J.~Tomko, editors, {\em Point {P}roceses and
  Queueing Problems (Colloquia Mathematica Societatis J\'{a}nos Bolyai 24)},
  volume~24, pages 257--275. North Holland, Amsterdam, 1980.

\bibitem{LSW01}
F.~A. Longstaff and E.S. Schwartz.
\newblock Valuing {A}merican {O}ptions by {S}imulation: a {S}imple
  {L}east-{S}quares {A}pproach.
\newblock {\em Review of Financial Studies}, 14(1):113--147, 2001.

\bibitem{lu2020}
K.W. Lu.
\newblock Calibration for {M}ultivariate {L}\'evy-{D}riven
  {O}rnstein-{U}hlenbeck {P}rocesses with {A}pplications to {W}eak
  {S}ubordination, 2020.
\newblock Available at www.arxiv.org.

\bibitem{SL2010}
E.~Luciano and P.~Semeraro.
\newblock {Multivariate {T}ime {C}hanges for {L}\'{e}vy {A}sset {M}odels:
  {C}haracterization and {C}alibration}.
\newblock {\em Journal of Computational and Applied Mathematics},
  233(1):1937--1953, 2010.

\bibitem{MadanSeneta90}
D.~B. Madan and E.~Seneta.
\newblock The {V}ariance {G}amma ({V}.{G}.) {M}odel for {S}hare {M}arket
  {R}eturns.
\newblock {\em The Journal of Business}, 63(4):511--24, 1990.

\bibitem{PellegrinoSabino15}
T.~Pellegrino and P.~Sabino.
\newblock Enhancing {L}east {S}quares {M}onte {C}arlo with {D}iffusion
  {B}ridges: an {A}pplication to {E}nergy {F}acilities.
\newblock {\em Quantitative Finance}, 15(5):761--772, 2015.

\bibitem{PSV20}
M.~Piccirilli, M.D. Schmeck, and T.~Vargiolu.
\newblock Capturing the {P}ower {O}ptions {S}mile by an {A}dditive {T}wo-factor
  {M}odel for {O}verlapping {F}utures {P}rices.
\newblock {\em Energy Economics}, 95:105006, 2021.

\bibitem{PoirotTankov2007}
J.~Poirot and Peter P.~Tankov.
\newblock Monte {C}arlo {O}ption {P}ricing for {T}empered {S}table ({CGMY})
  {P}rocesses.
\newblock {\em Asia-Pacific Financial Markets}, 13(4):327--344, 2006.

\bibitem{QDZ20}
Y.~Qu, A.~Dassios, and H.~Zhao.
\newblock Exact {S}imulation of {O}rnstein–{U}hlenbeck {T}empered {S}table
  {P}rocesses.
\newblock {\em Journal of Applied Probability}, 0(0), 2021.
\newblock Forthcoming.

\bibitem{ROSINSKI2007677}
Jan Rosinski.
\newblock Tempering {S}table {P}roceses.
\newblock {\em Stochastic Processes and their Applications}, 117(6):677 -- 707,
  2007.

\bibitem{Sabino20b}
P.~Sabino.
\newblock Exact {S}imulation of {V}ariance {G}amma {R}elated {OU} {P}roceses:
  {A}pplication to the {P}ricing of {E}nergy {D}erivatives.
\newblock {\em Applied Mathematical Finance}, 27(3):207--227, 2020.

\bibitem{Sabino20a}
P.~Sabino.
\newblock Forward or {B}ackward {S}imulations? {A} {C}omparative {S}tudy.
\newblock {\em Quantitative Finance}, 20(7):1213--1226, 2020.

\bibitem{cs20_1}
P.~Sabino and N.~Cufaro Petroni.
\newblock Gamma-related {O}rnstein–{U}hlenbeck {P}rocesses and their
  {S}imulation*.
\newblock {\em Journal of Statistical Computation and Simulation}, 0(0):1--26,
  2020.

\bibitem{Sato}
K.\ Sato.
\newblock {\em L\'evy {P}rocesses and {I}nfinitely {D}ivisible
  {D}istributions}.
\newblock Cambridge U.P., Cambridge, 1999.

\bibitem{SchwSchm00}
P.~Schwartz and J.E. Smith.
\newblock Short-term {V}ariations and {L}ong-term {D}ynamics in {C}ommodity
  {P}rices.
\newblock {\em Management Science}, 46(7):893--911, 2000.

\bibitem{ZhangOosterlee13_b}
B.~Zhang and C.W. Oosterlee.
\newblock An {E}fficient {P}ricing {A}lgorithm for {S}wing {O}ptions based on
  {F}ourier {C}osine {E}xpansions.
\newblock {\em Journal of Computational Finance}, 16(4):1--32, 2013.

\bibitem{ZhangOosterlee13_a}
B.~Zhang and C.W. Oosterlee.
\newblock Efficient {P}ricing of {E}uropean-style {A}sian {O}ptions under
  {E}xponential {L}{\'{e}}vy {P}rocesses based on {F}ourier {C}osine
  {E}xpansions.
\newblock {\em {SIAM} J. Financial Math.}, 4(1):399--426, 2013.

\end{thebibliography}
\end{document}